\documentclass[sigplan,screen]{acmart}

\AtBeginDocument{%
  \providecommand\BibTeX{{%
    \normalfont B\kern-0.5em{\scshape i\kern-0.25em b}\kern-0.8em\TeX}}}
    
\usepackage{graphicx}
\usepackage{subcaption}

\copyrightyear{2026}
\acmYear{2026}
\setcopyright{cc}
\setcctype{by}
\acmConference[ASPLOS '26]{Proceedings of the 31st ACM International Conference on Architectural Support for Programming Languages and Operating Systems, Volume 2}{March 22--26, 2026}{Pittsburgh, PA, USA}
\acmBooktitle{Proceedings of the 31st ACM International Conference on Architectural Support for Programming Languages and Operating Systems, Volume 2 (ASPLOS '26), March 22--26, 2026, Pittsburgh, PA, USA}
\acmPrice{}
\acmDOI{10.1145/3779212.3790250}
\acmISBN{979-8-4007-2359-9/2026/03}

\usepackage{multirow}
\usepackage[normalem]{ulem}
\usepackage{listings}
\usepackage{bm}
\usepackage{pifont}
\usepackage{hyperref}
\usepackage{float}
\usepackage{threeparttable}
\usepackage{xspace}
\usepackage{algorithmic}
\usepackage{algorithm}
\usepackage{xcolor}
\usepackage{makecell}
\usepackage{enumitem}
\useunder{\uline}{\ul}{}
\newcommand{\PHM}[1]{\vspace{.4em}\noindent\textbf{#1}\hspace{.5em}}
\newcommand{\SystemName}{\textsc{ZipServ}\xspace}

\newtheorem{theorem}{Theorem}[section]

\begin{document}

\title{\SystemName: Fast and Memory-Efficient LLM Inference\\with Hardware-Aware Lossless Compression}

\author{Ruibo Fan}
\orcid{0009-0006-7492-2069} 
\affiliation{%
  \institution{The Hong Kong University of Science and Technology (Guangzhou)}
  \city{Guangzhou}
  \country{China}
}
\email{ruibo.fan@connect.hkust-gz.edu.cn}

\author{Xiangrui Yu}
\orcid{0009-0005-2478-1512} 
\affiliation{%
  \institution{The Hong Kong University of Science and Technology (Guangzhou)}
  \city{Guangzhou}
  \country{China}
}
\email{xyu868@connect.hkust-gz.edu.cn}

\author{Xinglin Pan}
\orcid{0000-0002-1172-9935} 
\affiliation{%
  \institution{The Hong Kong University of Science and Technology (Guangzhou)}
  \city{Guangzhou}
  \country{China}
}
\email{xpan413@connect.hkust-gz.edu.cn}

\author{Zeyu Li}
\orcid{0009-0008-4381-0544} 
\affiliation{%
  \institution{The Hong Kong University of Science and Technology (Guangzhou)}
  \city{Guangzhou}
  \country{China}
}
\email{zli755@connect.hkust-gz.edu.cn}

\author{Weile Luo}
\orcid{0009-0007-2875-0056} 
\affiliation{%
  \institution{The Hong Kong University of Science and Technology (Guangzhou)}
  \city{Guangzhou}
  \country{China}
}
\email{wluo976@connect.hkust-gz.edu.cn}

\author{Qiang Wang}
\orcid{0000-0002-2986-967X} 
\affiliation{%
  \institution{Harbin Institute of Technology, Shenzhen}
  \city{Shenzhen}
  \country{China}
}
\email{qiang.wang@hit.edu.cn}

\author{Wei Wang}
\orcid{0000-0002-4585-4152} 
\affiliation{%
  \institution{The Hong Kong University of Science and Technology}
  \city{Hong Kong}
  \country{Hong Kong SAR}
}
\email{weiwa@cse.ust.hk}

\author{Xiaowen Chu}
\orcid{0000-0001-9745-4372} 
\affiliation{%
  \institution{The Hong Kong University of Science and Technology (Guangzhou)}
  \city{Guangzhou}
  \country{China}
}
\affiliation{%
  \institution{The Hong Kong University of Science and Technology}
  \city{Hong Kong}
  \country{Hong Kong SAR}
}
\email{xwchu@ust.hk}

\renewcommand{\shortauthors}{Ruibo Fan, et al.}
\renewcommand{\shorttitle}{ZipServ}

\sloppy
\begin{abstract}
Lossless model compression holds tremendous promise for alleviating the memory and bandwidth bottlenecks in 
bit-exact Large Language Model (LLM) serving. However, existing approaches often result in substantial inference slowdowns due to fundamental design mismatches with GPU architectures: at the kernel level, variable-length bitstreams produced by traditional entropy codecs break SIMT parallelism; at the system level, decoupled pipelines lead to redundant memory traffic.
We present \SystemName{}, a lossless compression framework co-designed for efficient LLM inference. \SystemName{} introduces \textit{Tensor-Core-Aware Triple Bitmap Encoding} (TCA-TBE), a novel fixed-length format that enables constant-time, parallel decoding, together with a \textit{fused decompression-GEMM} (ZipGEMM) kernel that decompresses weights on-the-fly directly into Tensor Core registers. This "\textit{load-compressed, compute-decompressed}" design eliminates intermediate buffers and maximizes compute intensity. Experiments show that \SystemName{} reduces the model size by up to 30\%, achieves up to 2.21$\times$ kernel-level speedup over NVIDIA's cuBLAS, and expedites end-to-end inference by an average of 1.22$\times$ over vLLM. \SystemName{} is the first lossless compression system that provides both storage savings and substantial acceleration for LLM inference on GPUs.
\end{abstract}

%
%


%

\maketitle
\section{Introduction}

\begin{figure}[tbp]
    \centering
\includegraphics[width=1.0\linewidth]{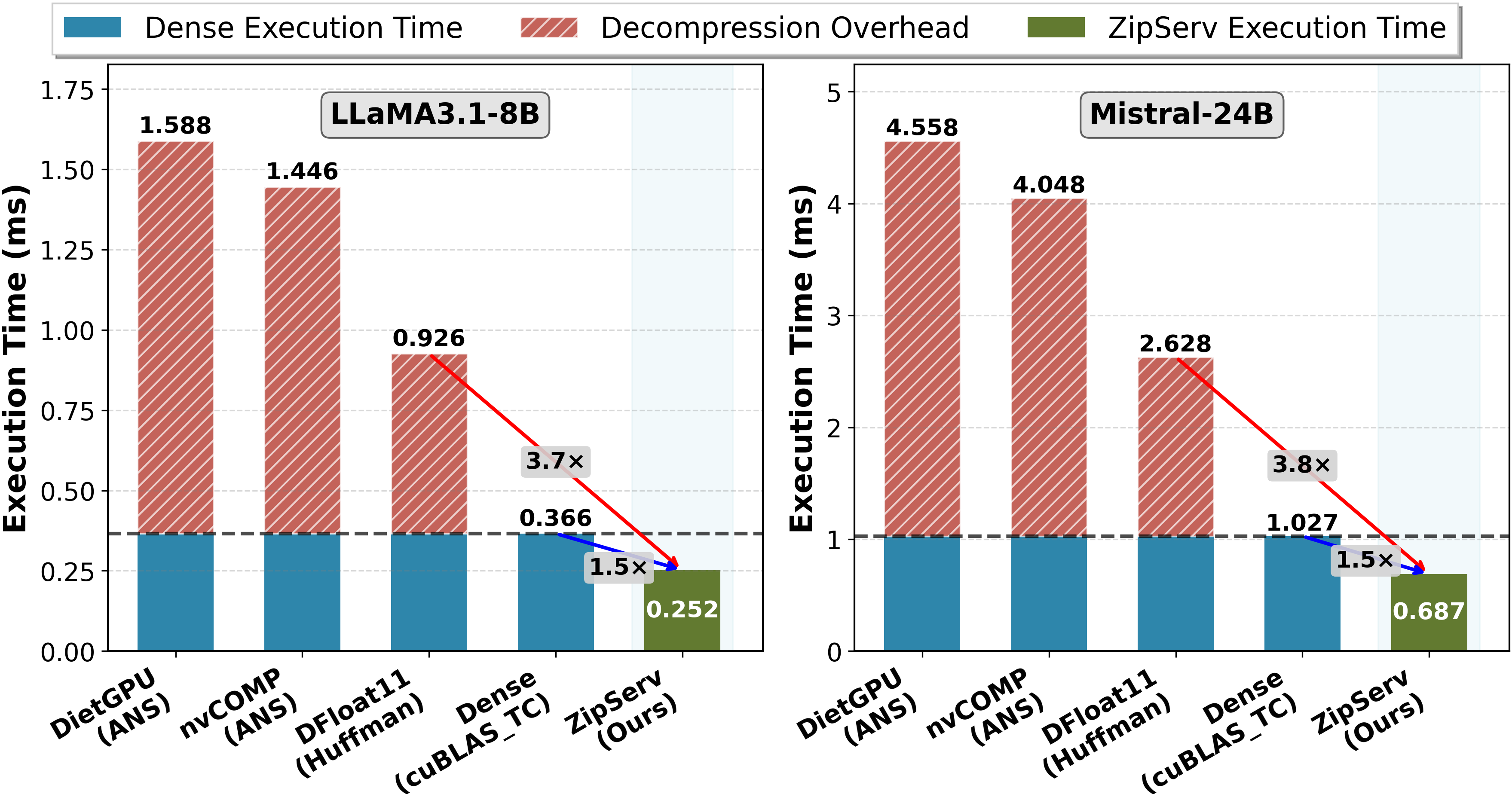}
    \caption{Execution time of lossless compression pipelines on NVIDIA L40S GPU with GateUp\_Proj layers.}
    \label{fig::duration}
\end{figure}

The transformative power of Large Language Models (LLMs) like GPT-4~\cite{openai2023gpt4}, LLaMA-3~\cite{dubey2024llama}, and Qwen-3~\cite{qwen3} is rooted in their massive scale~\cite{scaling_law, scaling_law_zhu}, enabling a new paradigm of AI applications~\cite{toolformer, chatbot_arena, Chain-of-Thought, allen_zhu_physics}. However, this immense scale creates significant deployment challenges, making GPU memory and memory bandwidth the primary bottlenecks for LLM serving, especially in resource-constrained environments.

Model compression offers a promising solution for efficient LLM deployment. Most existing approaches are \textit{lossy}, reducing size by approximating model weights via quantization (e.g., GPTQ~\cite{frantar2022gptq}, AWQ~\cite{lin2024awq}) or pruning (e.g., SparseGPT~\cite{Frantar2023SparseGPTML}). However, such approximations risk accuracy loss. 
For instance, aggressive 4-bit quantization (e.g., MXFP4) slashes accuracy from 56.0\% to 36.2\% on LiveCodeBench~\cite{quantization_study}, while even robust int8 quantization (GPTQ-int8) can cause up to 11.1\% loss in long-context reasoning (NOCHA)~\cite{mekala2025does}. 
These risks undermine reliability in safety-critical and user-facing settings, motivating approaches that guarantee bit-exact reproducibility and numerical integrity.

Lossless compression offers a compelling alternative by providing bit-exact model representation without accuracy loss. To date, its benefits have largely targeted storage and training workflows. For example, LMC~\cite{waddington2025lossless} and ZipNN~\cite{zipnn} employ Huffman~\cite{huffman2007method} to compress model checkpoints for efficient storage and distribution, while NeuZip~\cite{neuzip} and DietGPU~\cite{dietgpu} mitigate memory and communication overhead during training. Although recent efforts, notably DFloat11~\cite{dfloat11}, aim to extend these gains to inference, practical efficiency remains elusive. When integrated into serving pipelines, existing lossless techniques incur significant runtime overhead. As shown in Figure~\ref{fig::duration}, the decoupled decompression step alone takes 1.56–3.44$\times$ the time of the core inference computation. This overhead forces an \textit{unpleasant tradeoff} between memory efficiency and runtime efficiency. 

We contend that this tradeoff is not fundamental but arises from a mismatch between conventional compression algorithms and modern GPU architectures. The issue manifests at two levels. At the \textbf{kernel level}, traditional entropy codecs (e.g., Huffman~\cite{huffman2007method} or ANS~\cite{duda2015use}) produce variable-length bitstreams, whose decoding demands serialized, data-dependent operations. These are ill-suited to the lockstep, parallel SIMT execution model of GPU warps, resulting in severe control-flow divergence and compute underutilization. At the \textbf{system level}, most frameworks employ a \textit{decoupled} inference pipeline: weights are fully decompressed into a global-memory buffer before kernel consumption. This staged execution results in redundant, high-latency memory accesses, eroding compression-provided bandwidth savings and reducing arithmetic intensity during inference.

To rectify these fundamental algorithm-hardware mismatches, we present \textbf{\SystemName{}\footnote{Publicly available at \url{https://github.com/HPMLL/ZipServ_ASPLOS26.git}}}, the first lossless compression framework co-designed for high-performance LLM inference on GPUs. Our key observation is that the exponent bits of BFloat16 weights in LLMs exhibit a \textit{highly skewed, low-entropy distributions} in contemporary models. Exploiting this statistical redundancy, we propose \textit{Tensor-Core-Aware Triple Bitmap Encoding} (TCA-TBE), a fixed-length, bitmap-based weight format tailored to GPU architectures. Unlike variable-length entropy codecs, TCA-TBE enables constant-time, parallel decoding using lightweight bitwise operations, thereby eliminating control-flow divergence and aligning with the GPU's SIMT execution model. Paired with TCA-TBE, \SystemName{} devises a \textit{fused decompression-GEMM kernel} (ZipGEMM). Rather than decompressing weights into global memory as an intermediate step, ZipGEMM performs \textit{on-the-fly decoding}, delivering compressed weights directly into the register files that feed Tensor Core matrix multiplication units. This \textit{"load-compressed, compute-decompressed"} design eliminates intermediate buffers, reduces data movement, and maximizes the overlap between computation and memory access. By jointly addressing both the kernel-level and system-level mismatches, \SystemName{} transforms the theoretical storage savings of lossless compression into tangible performance gains on inference-optimized GPUs.

We demonstrate \SystemName{}'s effectiveness through comprehensive benchmarking against state-of-the-art lossless approaches, including DietGPU~\cite{dietgpu}, vendor-optimized nvCOMP~\cite{NVIDIA_nvcomp_2025}, and the Huffman-based DFloat11~\cite{dfloat11}. Compared to these baselines, which uniformly suffer significant runtime overhead, \SystemName{} consistently delivers substantial accelerations at both the kernel and system level on various inference-optimized GPUs, including RTX4090, L40S, and RTX5090. Our fused ZipGEMM achieves speedups of up to 2.21$\times$ over NVIDIA's cuBLAS, and up to 5.53$\times$ over DFloat11, the fastest lossless compression pipeline. These kernel-level improvements translate into an average 1.22$\times$ end-to-end speedup compared to leading systems like vLLM~\cite{vllm}. Our results demonstrate for the first time that when co-designed with hardware, lossless compression can provide both storage savings and substantial LLM inference acceleration.

The main contributions of this paper are as follows:
\begin{itemize}[noitemsep,nolistsep,,topsep=0pt,parsep=0pt,partopsep=0pt]
    \item We identify the fundamental mismatch between conventional entropy-based compression and GPU architectures, revealing both kernel- and system-level bottlenecks that hinder efficient inference.
    \item We propose TCA-TBE, a fixed-length, bitmap-based encoding tailored to SIMT execution and Tensor Core tiling, enabling constant-time, parallel decoding.
    \item We design ZipGEMM, a novel kernel that performs decompression on-the-fly directly into Tensor Core registers, eliminating intermediate memory buffers and maximizing compute intensity.
    \item We present and evaluate \SystemName{}, a lossless compressed LLM inference framework that achieves end-to-end speedups across diverse LLMs and GPUs, constituting the first practical evidence that lossless compression can directly accelerate LLM serving.
\end{itemize}

\section{Background}

\subsection{Transformer-Based LLMs}

Transformer-based LLMs~\cite{qwen3, mistral, dubey2024llama} are composed of stacked layers of multi-head attention, feed-forward networks~(FFNs), and normalization layers. During inference, computation proceeds \textit{autoregressively} in two phases: prefill and decode. The \textit{prefill} phase parallelizes computation over the input prompt, resulting in high arithmetic intensity due to large matrix multiplications operated over multiple tokens. On the contrary, the \textit{decode} phase generates tokens one at a time, where matrix multiplications involve only a single token per batch element. The decode phase, hence, suffers from reduced compute utilization and greater sensitivity to memory bandwidth. In both phases, the dominant operation is dense matrix multiplication: $Y = W X$, where
$W \in \mathbb{R}^{M \times K}$ is a learned weight matrix and $X \in \mathbb{R}^{K \times N}$ are activations, where
$M$ is the output dimension, $K$ is the hidden dimension, and $N$ is the number of tokens.

\subsection{BFloat16 Format}

\textbf{BFloat16 (BF16)}~\cite{kalamkar2019study} is a 16-bit floating-point format that has become the \emph{de facto} precision standard for LLM inference, balancing memory efficiency with numerical robustness. It is natively supported by major hardware accelerators, including NVIDIA Tensor Cores~\cite{luo2024benchmarking}, Google TPUs~\cite{jouppi2023tpu}, and Intel AMX~\cite{kim2024exploiting}, and is widely adopted in production-scale models, including LLaMA-3~\cite{dubey2024llama}, Qwen~\cite{qwen3}, and Mistral~\cite{mistral}. A BF16 number consists of 1 sign bit, 8 exponent bits, and 7 mantissa bits. Its numerical value is computed as:
\[
\text{BF16}(x) = (-1)^{\text{sign}} \times 2^{\text{exponent} - 127} \times (1.\text{mantissa}).
\]
This layout preserves the full exponent range of IEEE FP32 (1-8-23) while reducing mantissa precision. Compared to FP16 (1-5-10), BF16 offers a wider dynamic range, reducing vulnerability to overflows and underflows in large models.

\subsection{GPU Architecture and Tensor Core Execution}
Modern GPUs comprise multiple Streaming Multiprocessors (SMs), each with SIMT cores, Tensor Cores, registers, shared memory, and local caches. Threads are grouped into warps of 32, executing under the Single Instruction, Multiple Threads (SIMT) paradigm. Tensor Cores are specialized processors for high-throughput matrix multiplications. On recent NVIDIA architectures~\cite{nvidia_GA102whitepaper,nvidia_ada_whitepaper}, Tensor Cores support BF16 operands through the PTX-level \texttt{mma.sync.m16n8k16} instruction, which performs fused matrix multiply-accumulate (FMA) operations across small matrix tiles. A typical BF16 Tensor Core operation can be expressed as: $D_{\text{frag}} = A_{\text{frag}} \times B_{\text{frag}} + C_{\text{frag}}$,
where $A_{\text{frag}} \in \mathbb{R}^{16 \times 16}$, $B_{\text{frag}} \in \mathbb{R}^{16 \times 8}$, and $C_{\text{frag}} \in \mathbb{R}^{16 \times 8}$ is the FP32 accumulator fragment. This operation is executed at the warp level, where a group of 32 threads collaborate to compute the matrix multiplication. The input and output fragments are distributed across the entire warp. Each thread holds a specific subset of fragment elements in its registers, and the complete fragment is formed collectively. 

\section{Gaps and Opportunities}
\label{sec:gaps}

Lossless compression enables \textit{bit-exact} model representation but is rarely used for inference due to high runtime overheads stemming from a mismatch between traditional codecs and GPU architectures. This section quantifies compressibility in LLM weights and identifies key kernels and system-level bottlenecks that motivate our co-designed solution.

\begin{figure}[tbp]
  \centering
  \includegraphics[width=1.0\linewidth]{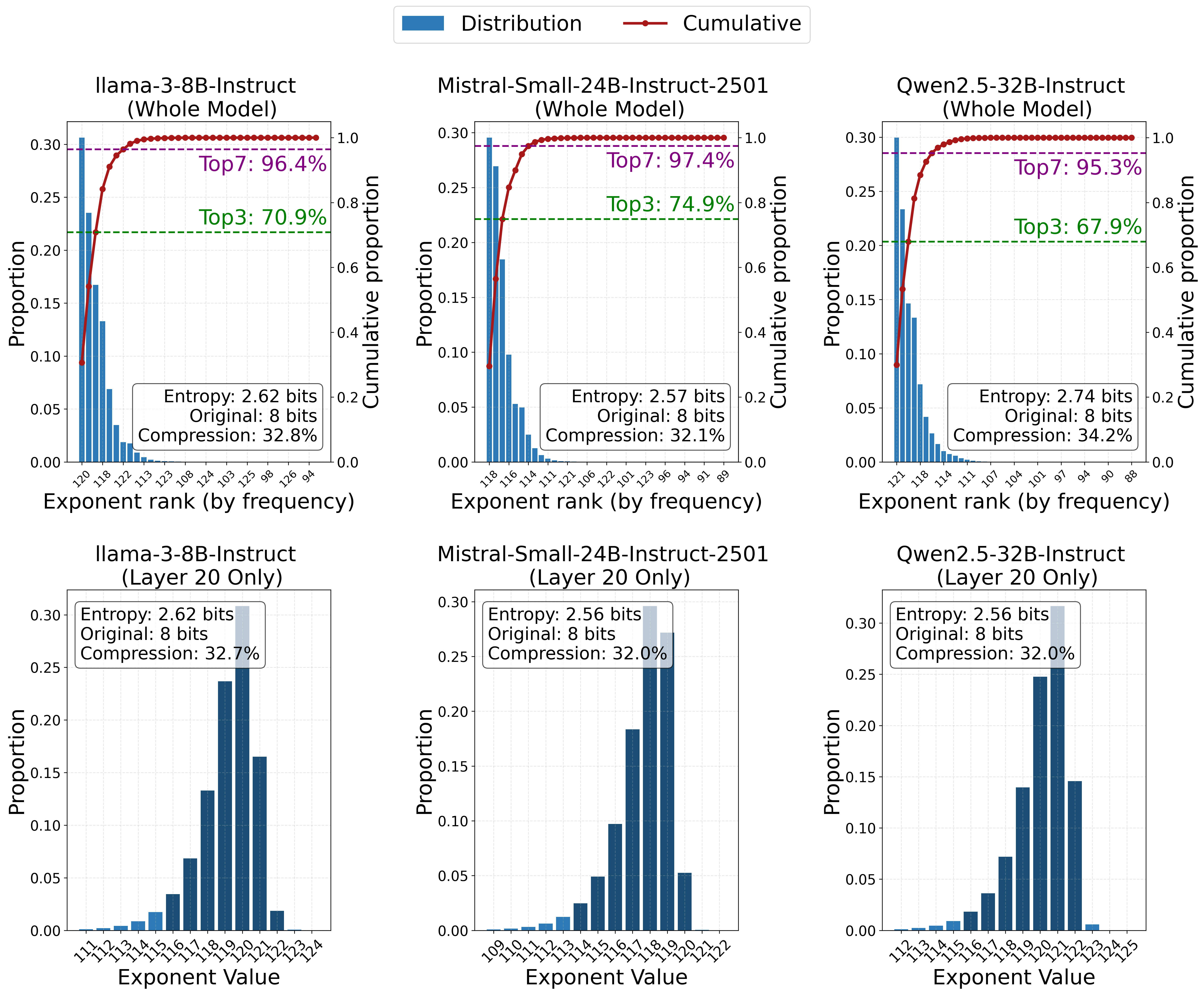}
\caption{Exponent bit distribution in LLM weights.}
\label{fig::dist}
\end{figure}

\subsection{Compressibility of BF16 Weights}\label{sec::3.1}
We analyzed the BF16 weights of leading LLMs, including Llama-3-8B-Instruct~\cite{dubey2024llama}, Mistral-Small-24B-Instruct-2501~\cite{mistral}, and Qwen2.5-32B-Instruct~\cite{yang2024qwen2.5}, and observed remarkable redundancy in their 8-bit exponent fields. As shown in Figure~\ref{fig::dist}, the exponent distributions are \emph{highly skewed}: the \textbf{top-3} most frequent exponents account for more than 67\% of all weights, and the \textbf{top-7} exponents cover over \textbf{95\%} (e.g., 96.4\% in Llama-3 and 97.4\% in Mistral-24B). The information entropy of the exponent field is only \textbf{2.57--2.74 bits}, far below its 8-bit allocation, implying a theoretical lossless compression ratio of about 1.51$\times$ (16/10.6) for BF16 values. These findings are consistent with prior works \cite{dfloat11, neuzip, zipnn, yubeaton2025huffllm, waddington2025lossless}. We further scrutinized this redundancy across 3,875 weight matrices from four LLM families (Gemma-3, Mistral, Qwen2.5, and LLaMA-3.1), revealing a critical structural property: \textit{exponent contiguity}. In 99.6\% of these matrices, the top-7 most frequent exponents form a numerically contiguous sequence (i.e., ${e^\star, \dots, e^\star+6}$). Consequently, a simple contiguous window covers 97.1\% of all weights on average, approaching the information-theoretic limit. In Appendix~\ref{sec:theory_appendix}, we prove that this is not coincidental but an intrinsic property of LLMs. This contiguity is the cornerstone of \SystemName{}. It obviates the need for complex, hardware-unfriendly variable-length codecs (e.g., Huffman) in favor of a \textbf{fixed-length}, base-plus-offset representation. This insight directly enables our Tensor-Core-Aware Triple Bitmap Encoding (TCA-TBE) and its implicit lookup mechanism described in \S\ref{sec::decomp}.

\begin{figure}[tbp]
\centering
\includegraphics[width=0.9\linewidth]{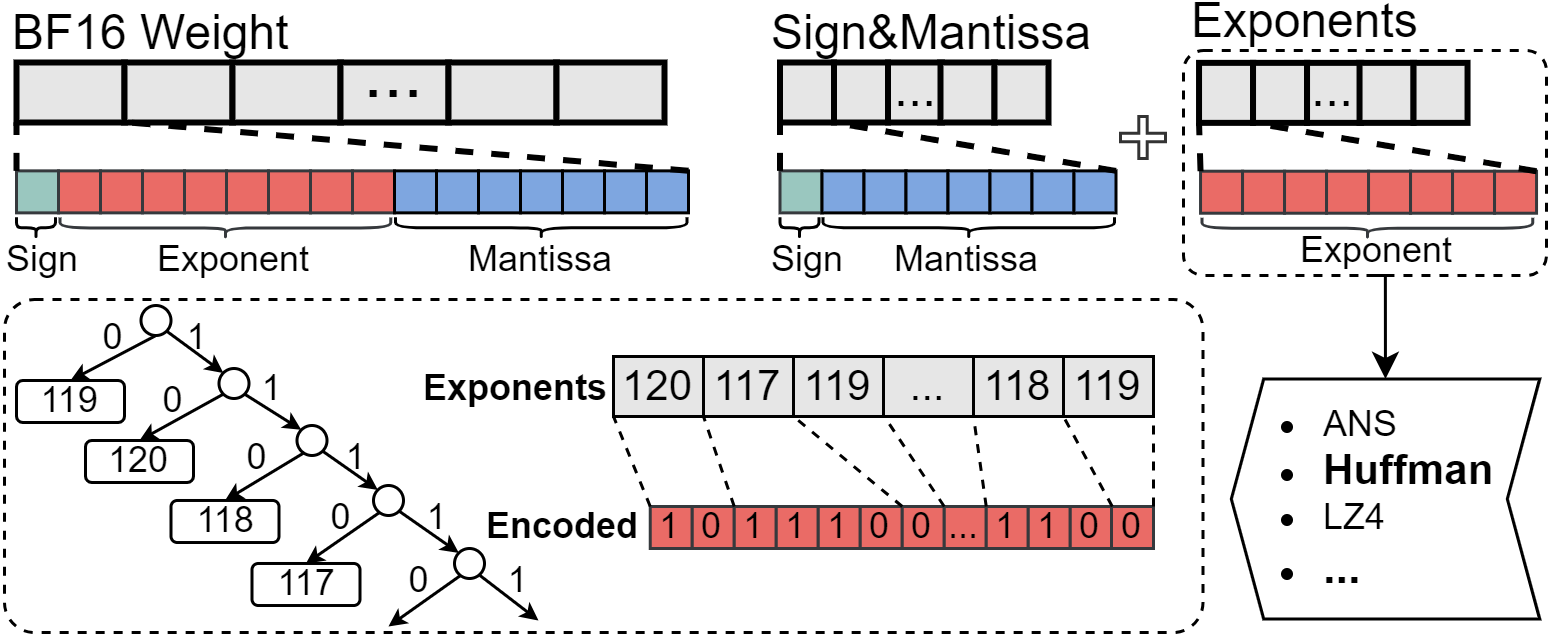}
\caption{Existing Lossless Compression for BF16 Weights. Illustrated with Huffman Encoding.}
\label{fig::lossless}
\end{figure}

\subsection{Kernel-Level Architectural Mismatch}
\label{subsec:kernel-level gaps}
Existing methods exploit the exponent redundancy of BF16 weights by applying entropy coding to the exponent stream. 
For example, DFloat11 uses Huffman coding~\cite{dfloat11}, while DietGPU employs Asymmetric Numeral Systems (ANS)~\cite{dietgpu}. 
As shown in Figure~\ref{fig::lossless}, these approaches produce a compressed bitstream with \textit{variable-length} symbols 
depending on their statistical frequency. However, this bitstream must be decompressed \textit{sequentially} to correctly recover each exponent, 
which fundamentally conflicts with the lockstep, massively parallel SIMT execution model of modern GPUs.

To illustrate this mismatch, we examine the three-stage decompression pipeline in DFloat11~\cite{dfloat11}. \ding{182} \textbf{Bitstream Partitioning.} The bitstream is split into chunks for parallel thread processing. However, because variable-length symbols cross chunk boundaries, threads cannot operate independently but require additional metadata to locate valid symbol start points, introducing overhead and disrupting parallel execution. \ding{183} \textbf{Symbol Extraction.} Threads use hierarchical lookup tables (LUTs) for symbol decoding---a data-dependent operation. When warp threads encounter different symbol lengths, faster threads stall for slower ones, causing divergence and underutilization of GPU resources. \ding{184} \textbf{Pointer Advancement.} After symbol decoding, each thread advances its bit pointer by the symbol's length, which
is only known after the lookup completion. This inherently serializes the decoding loop and sacrifices opportunities for instruction-level parallelism. 
Our evaluation shows that on L40S GPUs, even highly optimized decompressors (e.g., ANS-based DietGPU and Huffman-based DFloat11) achieve only 43.7\% and 76.5\% of peak memory bandwidth, respectively. This inefficiency exposes a fundamental algorithm-hardware mismatch: entropy coding is inherently data-dependent, while efficient GPU execution desires regular, uniform parallelism. 

\begin{figure}[tbp]
  \centering
  \includegraphics[width=0.8\linewidth]{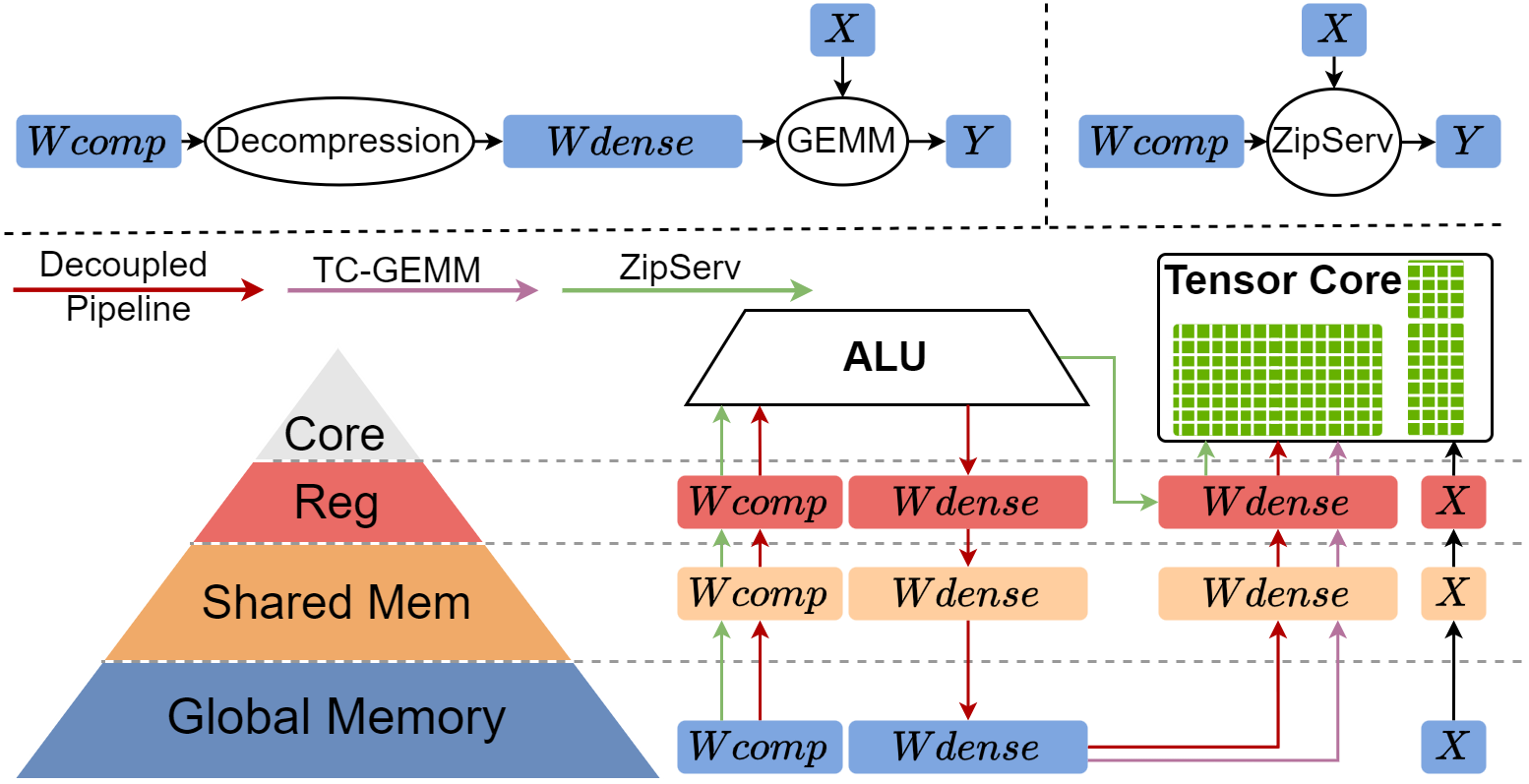}
\caption{Existing lossless compression inference pipeline.}
\label{fig::compare}
\end{figure}

\subsection{Inefficiency of Decoupled Inference Pipeline}

The architectural inefficiency of entropy-coded decoding is found not only at the kernel level, 
but also at the system pipeline level for LLM inference. In mainstream approaches, decompression is 
performed as a separate, decoupled preprocessing stage (see Figure~\ref{fig::compare}): it
materializes the entire decompressed weights in global memory first and then passes it to the compute kernels.
This decoupled pipeline design leads to redundant data transfers, undermining the benefits of compression, particularly 
in bandwidth-constrained environments. We analytically quantify its inefficiency using the Roofline model, focusing
on Compute Intensity (CI).

\PHM{Compute Intensity.} CI measures the number of floating-point operations (FLOPs) performed per byte read from global memory. 
For a typical BF16 GEMM operation $Y_{M \times N} = W_{M \times K} X_{K \times N}$, the compute intensity is:
\begin{equation}
\footnotesize
 CI_{GEMM} = \frac{MNK}{MK + KN + MN}.
\end{equation}
In the decoupled pipeline scenario, assuming an average compression ratio (CR) of 1.51 (\S\ref{sec::3.1}), the CI becomes:
\begin{equation}
\footnotesize
CI_{\text{Decoupled}} = \frac{2MNK}{MK\left(\frac{2}{\text{CR}} + 4\right) + 2(KN + MN)} \approx \frac{MNK}{2.66 MK + KN + MN}.
\end{equation}


\PHM{Roofline Model Analysis.}
Figure~\ref{fig::roofline} illustrates the Roofline analysis on an NVIDIA RTX4090. During the decode stage, both the standard GEMM and the decoupled pipeline operate in the memory-bound regime, where performance scales linearly with CI. However, our analysis highlights a pronounced penalty for the decoupled approach: the additional memory traffic required to materialize intermediate decompressed weights significantly reduces CI. Specifically, for a weight matrix of size $M=K=4096$, the decoupled pipeline exhibits a CI degradation of 62.3\%, 62.2\%, 62.0\%, and 61.7\% relative to standard GEMM for batch sizes of 8, 16, 32, and 64, respectively.

\begin{figure}[tbp]
  \centering
  \includegraphics[width=1.0\linewidth]{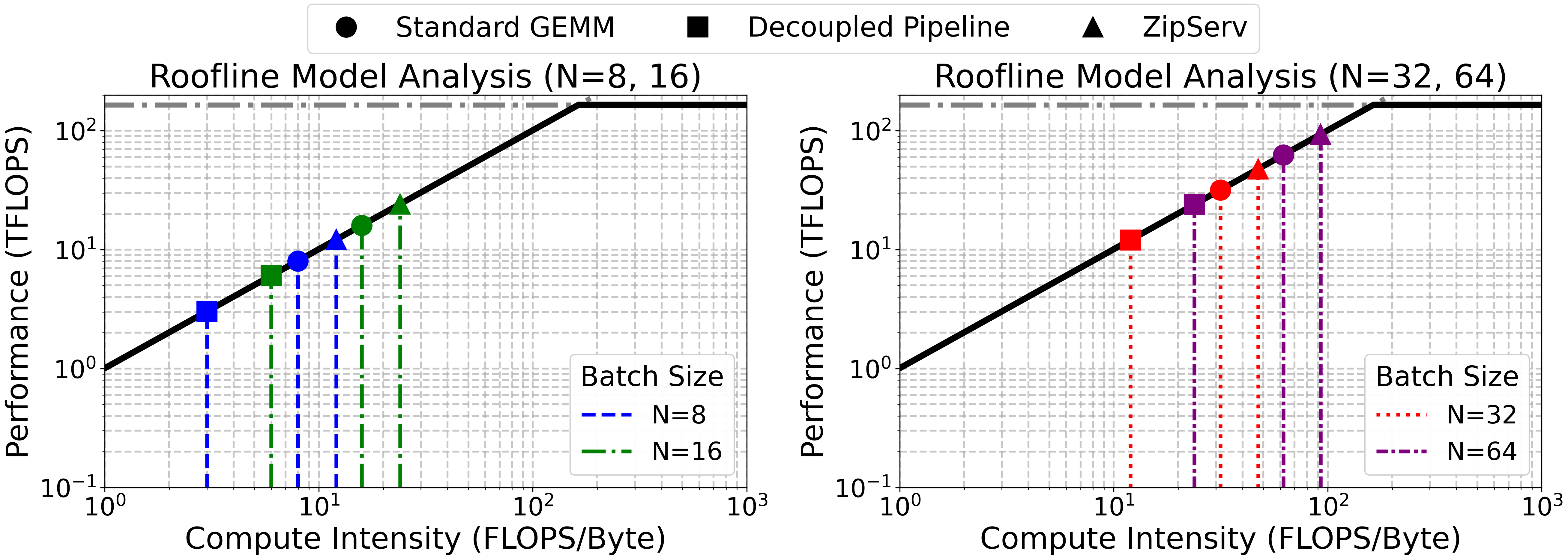}
\caption{Roofline analysis.}
\label{fig::roofline}
\end{figure}

\PHM{\SystemName{}'s Fused Design.} 
The inefficiency of decoupled pipelines arises directly from staging decompressed weights in global memory.
\SystemName{} addresses this by introducing a fused decompression-GEMM kernel that directly fetches 
compressed weights from DRAM and decompresses them on-the-fly into register files, which immediately feed the Tensor Core. 
This approach effectively increases CI to
\begin{equation}
\footnotesize
CI_{\SystemName{}} = \frac{2MNK}{MK \cdot \frac{2}{\text{CR}} + 2(KN + MN)} \approx \frac{MNK}{ 0.66MK + KN + MN}.
\end{equation}
Revisiting the Roofline model in Figure~\ref{fig::roofline}, \SystemName{}'s fused execution ($CI_{\SystemName{}}$) demonstrates a 
substantial improvement, achieving even higher CI (approximately  50\%) than the uncompressed GEMM baseline. This benefit, most pronounced in memory-bound regimes, leads to linear speedups relative to the compression ratio, translating information-theoretic redundancy into wall-clock acceleration.

\begin{figure*}[tbp]
  \centering
\includegraphics[width=0.9\linewidth]{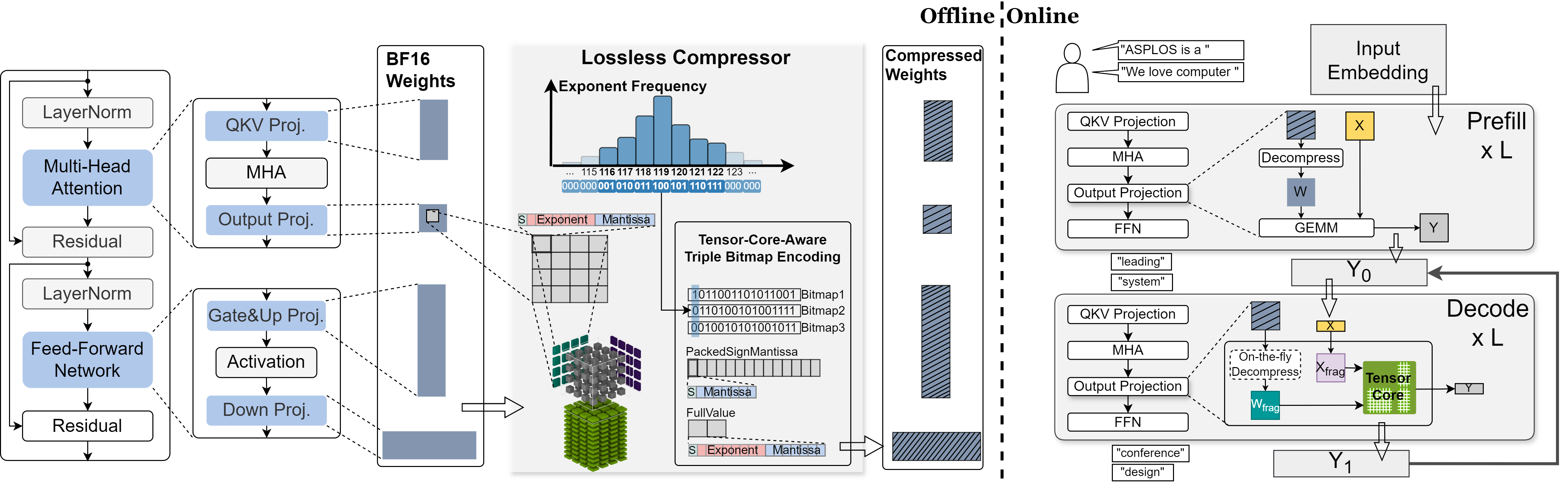}
\caption{Overview of \SystemName{}. \SystemName{} features an offline lossless compressor (left) and an online inference engine (right).}
\label{fig::overview}
\end{figure*}

\section{Design of \SystemName{}}
\label{sec:design}

Our earlier analysis identifies both kernel-level and system-level sources of inefficiency that hinder the decoding of lossless compression in LLM inference. In this section, we present \SystemName{}, a lossless compression system co-designed for storage efficiency and \textit{fast, bit-exact} LLM inference.

\subsection{Overview and Workflow}

As illustrated in Figure~\ref{fig::overview},  
\SystemName{} consists of two main components: an \textit{offline compressor}, which transforms BF16 model weights into a parallelization-friendly compressed representation, and an \textit{online inference engine}, responsible for efficient decoding and computation at runtime.


\PHM{Offline Compressor.}
At the core of the offline compressor is the \textit{Tensor-Core-Aware Triple Bitmap Encoding} (TCA-TBE), a fixed-length, bitmap-based compression format designed to enable \textit{parallel decoding} via GPU SIMT execution and Tensor Core–accelerated GEMM operations. As outlined in \textbf{Algorithm~\ref{alg:compression}}, given a model, the compressor first profiles the exponent distribution of each layer’s weights. Instead of selecting arbitrary frequent exponents, it identifies a window of $k$ \textit{numerically consecutive} exponent values (typically $k=7$) that maximizes coverage of the weight distribution. The compressor records the value immediately preceding this range as the \texttt{BaseExp} (i.e., $\min(\text{range}) - 1$).
Using this range, the compressor encodes the entire weight matrix into the TCA-TBE representation. Each $8 \times 8$ tile of weights is converted into three 64-bit bitmaps and two compact value buffers: one for high-frequency values falling within the selected exponent range (storing only the sign and mantissa relative to \texttt{BaseExp}), and another for outliers in full BF16 precision. The resulting compressed model is then loaded onto the GPU, ready for serving.

\PHM{Online Inference Engine.} 
The inference engine employs a stage-aware strategy that adapts the execution pipeline for the prefill and decode phases, all on the unified TCA-TBE format.
During the compute-bound \textbf{prefill stage}, the engine performs \textit{decoupled execution}: a dedicated decompression kernel decompresses the weights into global memory first, followed by the prefill computation. This approach allows high-throughput GEMM to effectively amortize the decompression overhead.
In the memory-bound \textbf{decode stage}, the engine switches to a fused decompression-GEMM kernel (ZipGEMM). ZipGEMM enables a \emph{``load-compressed, compute-decompressed''} execution model, where weights are decompressed on-the-fly directly into Tensor Core registers. This eliminates redundant data transfers and maximizes compute intensity for each token generation. These two specialized execution paths deliver near-optimal inference performance.

\begin{algorithm}[tbp]
  \caption{\SystemName{} Offline Compressor (TCA-TBE)}\label{alg:compression}
  \begin{algorithmic}[1]
      \REQUIRE Weight Matrix $\mathcal{W}$, Tile Size $T=8\times8$
      \ENSURE Bitmaps $\mathcal{B}_{1..3}$, High-Freq Buffer $\mathcal{H}$, Fallback Buffer $\mathcal{L}$, BaseExp $e_{base}$
      
      \STATE \textcolor{blue}{$\triangleright$ Phase I: Global Exponent Analysis}
      \STATE $Hist \leftarrow \textsc{ComputeExponentHistogram}(\mathcal{W})$
      \STATE $E_{top} \leftarrow \textsc{SelectTop7ConsecutiveExponents}(Hist)$
      \STATE $e_{base} \leftarrow \min(E_{top}) - 1$ \hfill \textcolor{blue}{$\triangleright$ Set base for implicit lookup}
      
      \STATE \textcolor{blue}{$\triangleright$ Phase II: Tile Encoding}
      \FOR{\textbf{each} tile $t \in \mathcal{W}$}
          \STATE Initialize local bitmaps $b_1, b_2, b_3 \leftarrow 0$
          \FOR{$i = 0$ \textbf{to} $63$}
              \STATE $w \leftarrow t[i]$; \quad $e \leftarrow w.exponent$
              \IF{$e \in E_{top}$}
                  \STATE $c \leftarrow e - e_{base}$ \hfill \textcolor{blue}{$\triangleright$ Compute 3-bit code $c \in [1, 7]$}
                  \STATE $b_1[i] \leftarrow c_0; \ b_2[i] \leftarrow c_1; \ b_3[i] \leftarrow c_2$ \hfill \textcolor{blue}{$\triangleright$ Set bits}
                  \STATE $\mathcal{H}.\textsc{Push}(\textsc{Pack}(w.sign, w.mantissa))$
              \ELSE
                  \STATE $\mathcal{L}.\textsc{Push}(w)$ \hfill \textcolor{blue}{$\triangleright$ Store full precision fallback}
              \ENDIF
          \ENDFOR
          \STATE Store $b_1, b_2, b_3$ to global $\mathcal{B}_{1..3}$
      \ENDFOR
  \end{algorithmic}
\end{algorithm}

\begin{figure}[tbp]
  \centering
  \includegraphics[width=0.9\linewidth]{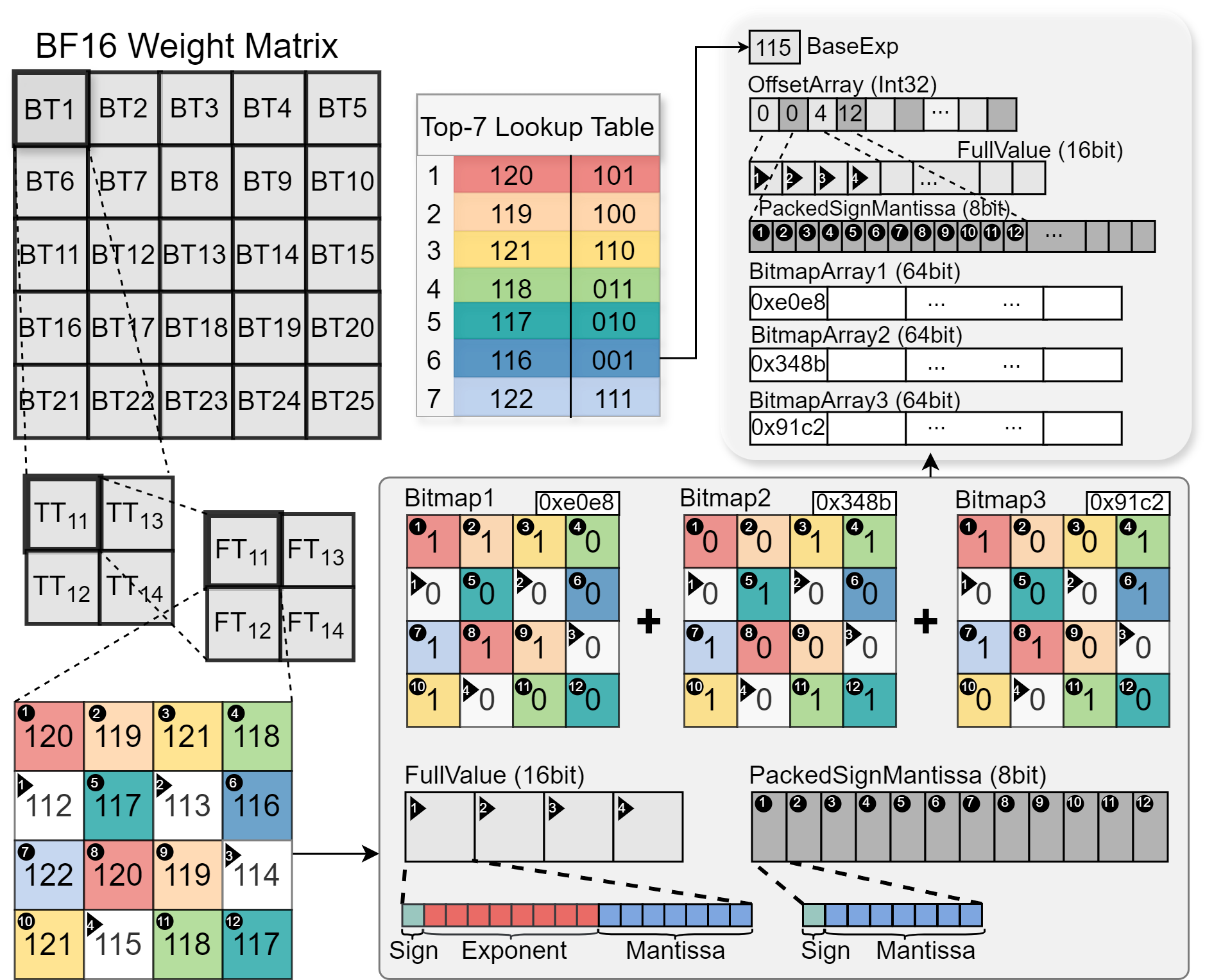}
\caption{Tensor-Core-Aware Triple Bitmap Encoding. The $4 \times 4$ FragTile shown is illustrative; the actual size is $8 \times 8$.}
\label{fig::tcabme}
\end{figure}

\subsection{Tensor-Core-Aware Triple Bitmap Encoding} \label{sec::tca-bme}
\SystemName{} is built on top of a novel Tensor-Core-Aware Triple Bitmap Encoding (TCA-TBE) scheme. It is designed to minimize the weight memory footprint while enabling efficient parallel decoding on GPUs. In contrast to existing variable-length bitstream-based entropy codecs, TCA-TBE employs a fixed-length, tile-structured representation that ensures constant-time, thread-local decompression. Its data layout is carefully aligned with Tensor Core tiling and register-level operand distribution, allowing the decompressed weights to be consumed directly by the \texttt{mma.sync} instruction. The core of TCA-TBE is a fixed-length \textbf{3-bit codeword} assigned to each weight element, representing one of eight possible states (\texttt{000}–\texttt{111}). During offline compression, \SystemName{} profiles the exponent histogram of a weight matrix and identifies the top-7 most frequent exponent values and maps them to codewords \texttt{001}–\texttt{111}. The special codeword \texttt{000} serves as a fallback, designating weights whose exponent falls outside the top-7, which are then stored in full precision. 

\PHM{The Choice of Codeword Length.} We choose the 3-bit codeword because it achieves a near-optimal compression ratio by leveraging the highly skewed exponent distributions observed in contemporary LLMs. To quantify this design choice, we calculate the expected per-element storage cost as:
\[
AverageBits(n) = r_n \cdot (n + 8) + (1 - r_n) \cdot (n + 16),
\]
where $n$ is the codeword length and $r_n$ is the proportion of weights covered by the top $2^n - 1$ exponent values. As shown in \S\ref{sec::3.1}, $r_3 \approx 0.96$, yielding an average of 11.3 bits per element, which approaches the theoretical lower bound (8+2.6=10.6 bits) and offers clear advantages over 2-bit (12.4 bits) and 4-bit (12.1 bits) codewords. Besides, the 3-bit encoding yields a compact 7-entry codebook, enabling decoding via a simple table lookup. This requires only a handful of bitwise operations per thread, which can be efficiently performed with warp-synchronous Tensor Core pipelines.

\PHM{Decoupled Triple Bitmap Layout.} To maximize decoding efficiency on SIMT architectures, TCA-TBE implements a \textit{decoupled triple bitmap layout} rather than packing codewords into a dense bitstream. Conventional bitstreams are inefficient on GPUs because packing non-byte-aligned codes (e.g., 3-bit) forces codewords to span memory word boundaries. This necessitates complex logic for non-aligned accesses and introduces data-dependent branching, which in turn causes thread divergence that severely degrades SIMT throughput.

TCA-TBE avoids these bottlenecks by decomposing the 3-bit codewords for each $8 \times 8$ weight tile into three independent 64-bit bitmaps, with each bitmap representing a single bit-plane (Figure \ref{fig::tcabme}). This design enables two benefits. First, it guarantees coalesced memory accesses, as each bitmap is a contiguous 64-bit word, naturally aligned to native memory boundaries. Second, it enables branch-free decoding. All threads in a warp follow an identical execution path, aligning with the SIMT model on modern GPUs.

\PHM{Hierarchical Tiling Design.}
TCA-TBE adopts a three-level hierarchical tiling scheme that partitions the weight matrix according to the architectural granularity of modern GPUs. \ding{182} FragTile (FT): The base unit is an $8 \times 8$ tile, matching the smallest operand fragment of Tensor Core instruction. \ding{183} TensorCoreTile (TT): Each $16 \times 16$ tile is composed of a $2 \times 2$ grid of FragTiles. This size aligns with the operand dimensions (\texttt{m=16}, \texttt{k=16}) required by PTX-level Tensor Core \texttt{mma} instructions (\texttt{mma.m16n8k16}). \ding{184} BlockTile (BT): At the coarsest level, a $64 \times 64$ tile aggregates multiple TensorCoreTiles and is processed cooperatively by a thread block. The FragTiles within a TensorCoreTile are stored in \textit{column-major order}, mirroring the operand register layout (e.g., Ra0–Ra3) expected by Tensor Core instructions. This design eliminates the need for runtime coordinate transformation, reducing instruction overhead. Each $8 \times 8$ FragTile is encoded using five buffers. \ding{182} Three 64-bit bitmaps, each representing one bit-plane of the 3-bit codewords. \ding{183} A \texttt{PackedSignMantissa} buffer, which holds the compact 8-bit representation (sign and mantissa) of weights whose exponents fall within the top-$k$ frequent classes. \ding{184} A \texttt{FullValue} buffer, which stores full-precision BF16 values for weights not covered by the exponent codebook. At the matrix level, TCA-TBE organizes these buffers into four contiguous global arrays, each nested according to the tiling hierarchy. In addition, an \texttt{Offset} array records the starting offset of each GroupTile within the \texttt{PackedSignMantissa} and \texttt{FullValue} arrays.

\begin{figure}[tbp]
  \centering
  \includegraphics[width=0.95\linewidth]{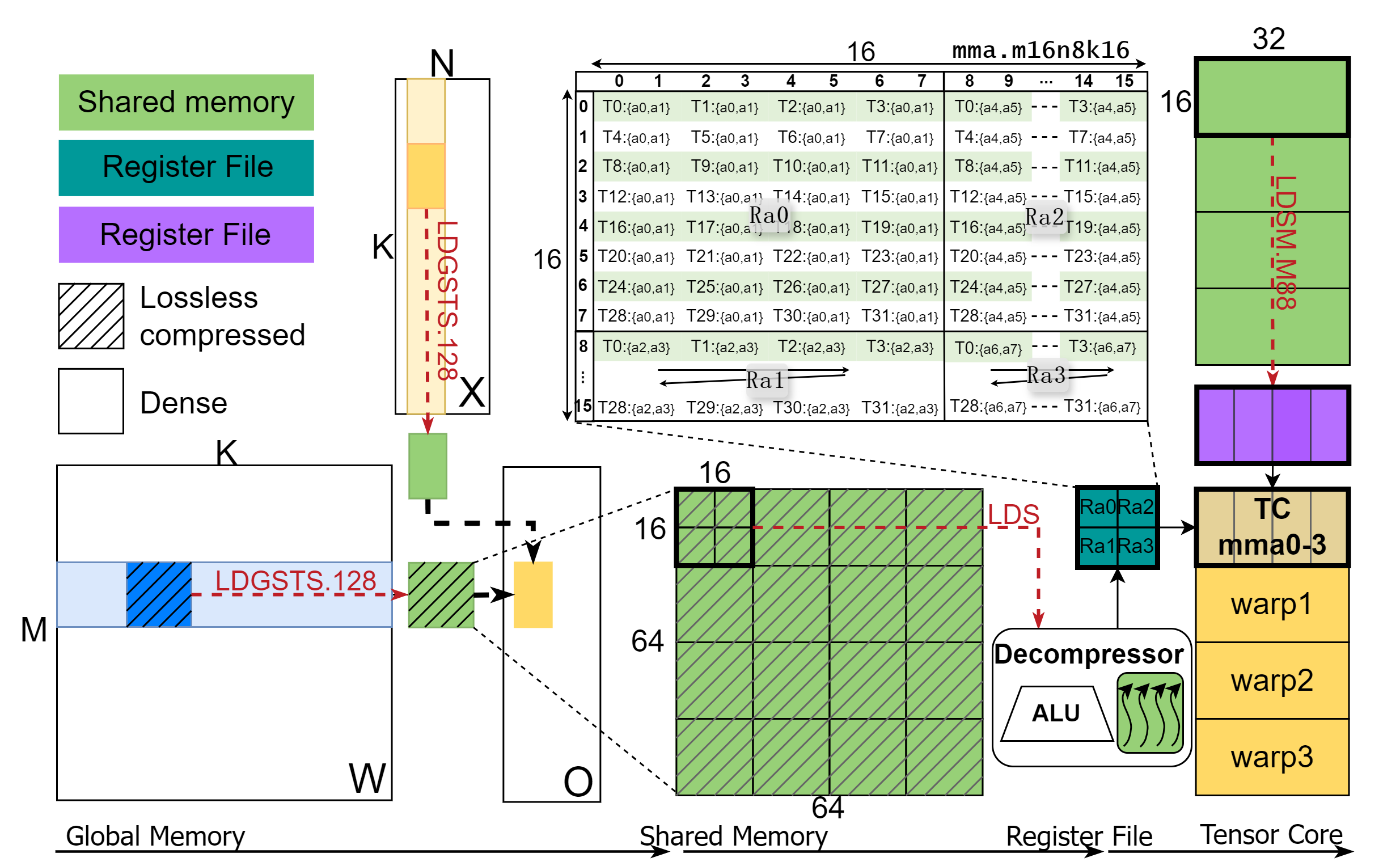}
\caption{Data movement and instruction pipeline.}
\label{fig::move}
\end{figure}

\subsection{Fused ZipGEMM Kernel Design}
\label{sec:ZipGEMM}

TCA-TBE's SIMT-friendly design opens up new opportunities for high-throughput decoding. To achieve this, \SystemName{} fuses decompression and matrix multiplication into a single kernel, \textbf{ZipGEMM}, that fetches weights from global memory in a compact TCA-TBE format and decompresses them just-in-time during computation. ZipGEMM enables a \textbf{\textit{load-compressed, compute-decompressed}} execution model, substantially reducing the memory bandwidth requirement for each token generation in the
decode stage (see Figure~\ref{fig::roofline}). 

\begin{figure*}[tbp]
  \centering
  \includegraphics[width=0.85\linewidth]{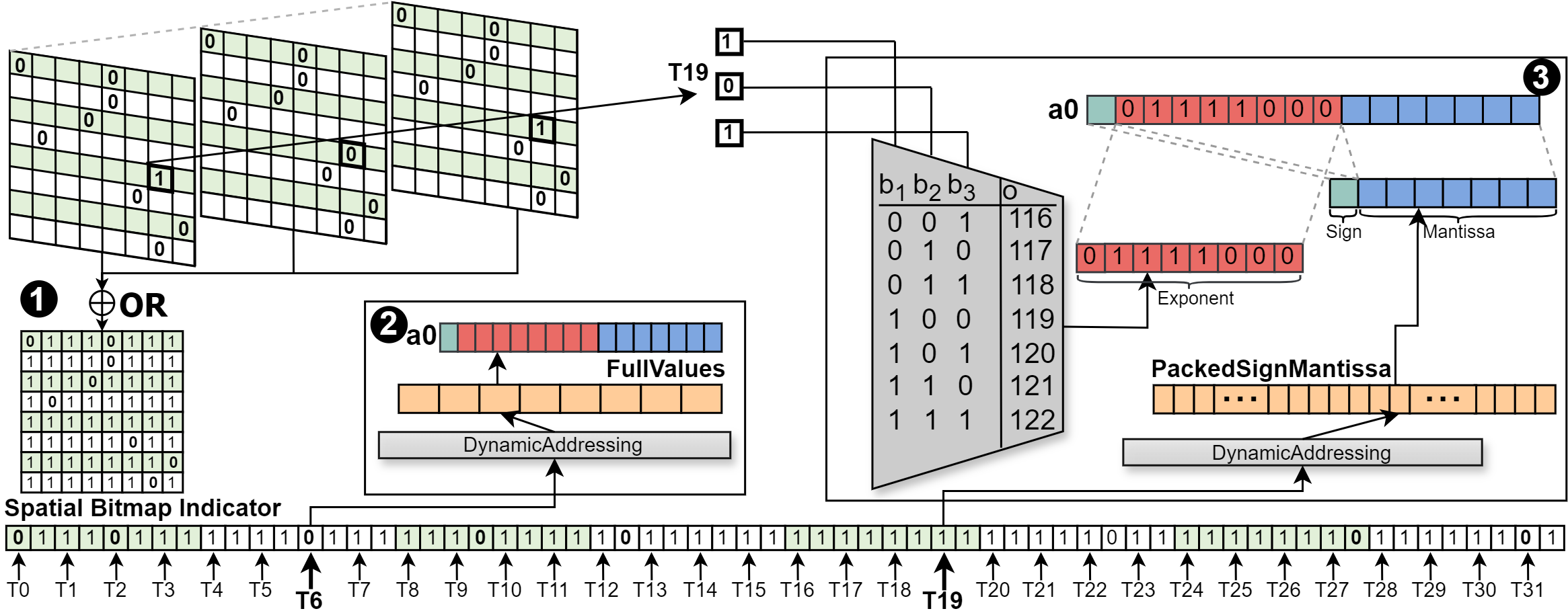}
\caption{The Decompressor Design.}
\label{fig::bitmapdecode}
\end{figure*}
\subsubsection{Kernel Workflow}
\label{sec:ZipGEMM-workflow}

Figure~\ref{fig::move} illustrates the workflow of the ZipGEMM kernel. Based on a split-K tiling architecture, each thread block iteratively processes the $K$ dimension in chunks. In each iteration, the kernel proceeds through four coordinated stages. \ding{182} Tile Loading. Threads cooperatively load the compressed weight tile and the corresponding activation tile from global memory into shared memory, with asynchronous and vectorized memory instructions (i.e., \texttt{LDGSTS.128}) to bypass the L1 cache and improve global memory bandwidth utilization. The \texttt{PackedSignMantissa} and \texttt{FullValue} arrays within each tile are padded offline to ensure 128-bit alignment. \ding{183} Warp-Level Decoding. Each warp independently decompresses the compressed weight from shared memory. The decompressor reconstructs the original BF16 values in a layout compatible with Tensor Core consumption, utilizing lightweight ALU operations and avoiding shared memory round-trips.
\ding{184} Activation Register Transfer. The activation tile is moved from shared memory into registers using the \texttt{LDSM.M88} instruction, which enables a warp to load a $16 \times 16$ tile and arrange it in the layout required for Tensor Cores. \ding{185} Tensor Core Computation. Once both decompressed weights and activations reside in registers, the warp performs Tensor Core \texttt{mma} instructions. The execution path closely mirrors the standard cuBLAS GEMM kernels, while operating directly on compressed representations and reducing global memory accesses.

\subsubsection{Efficient Decompressor}\label{sec::decomp}
ZipGEMM incorporates an efficient Decompressor that enables thread-local reconstruction of compressed weights directly within the register file. The core principle of the Decompressor is that \textit{each thread independently decompresses the elements required for the proper Tensor Core fragment layout}. Specifically, as shown in Figure \ref{fig::tcabme}, the fragment layout requires that thread~$i$’s \texttt{.bf16x2} register (e.g., \texttt{Ra0}) holds the values at positions $2i$ and $2i+1$ within the $8 \times 8$ tile, denoted as $a_0$ and $a_1$ respectively. Since each element is encoded in one of two states---either as a high-frequency fixed-length code or as a fallback full-precision value---and these states are distributed in an unstructured manner, the decompressor solves a sparse, non-uniform spatial reconstruction problem. Two challenges arise in this context. First, each thread must efficiently determine the state of its assigned element (compressed or fallback). Second, each thread should recover the original BF16 representation in a deterministic, SIMT-friendly manner. To this end, \SystemName{}’s Decompressor is structured into three tightly integrated stages: spatial bitmap indicator, dynamic addressing, and fast exponent reassembly (see Figure \ref{fig::bitmapdecode} and Algorithm~\ref{alg:decompression}).

\begin{algorithm}[t]
  \caption{ZipGEMM Thread-Local Decompression}\label{alg:decompression}
  \begin{algorithmic}[1]
      \REQUIRE Bitmaps $\mathcal{B}_{1..3}$, Buffers $\mathcal{H}, \mathcal{L}$, BaseExp $e_{base}$, LaneID~$l$
      \ENSURE Register pair $R$ containing two BF16 values
      
      \STATE \textcolor{blue}{$\triangleright$ Step 1: Spatial Indicator Construction}
      \STATE $\mathcal{M} \leftarrow \mathcal{B}_1 \lor \mathcal{B}_2 \lor \mathcal{B}_3$
      \STATE \textcolor{blue}{$\triangleright$ Step 2: Parallel Element Decompression}
      \FOR{$k \in \{0, 1\}$}
          \STATE $p \leftarrow 2 \cdot l + k$ \hfill \textcolor{blue}{$\triangleright$ Global position in $8 \times 8$ tile}
          \STATE $mask \leftarrow (1 \ll p) - 1$
          \STATE $idx_{\mathcal{H}} \leftarrow \textsc{Popc}(\mathcal{M} \ \& \ mask)$ \hfill \textcolor{blue}{$\triangleright$ Calculate index}
          
          \IF{$(\mathcal{M} \gg p) \ \& \ 1$}
              \STATE \textcolor{blue}{$\triangleright$ Case A: High-Frequency Path}
              \STATE $val \leftarrow \mathcal{H}[\text{start}_{\mathcal{H}} + idx_{\mathcal{H}}]$ \hfill \textcolor{blue}{$\triangleright$ Fetch Sign + Mantissa}
              \STATE  \textcolor{blue}{$\triangleright$Reconstruct 3-bit code}
              \STATE $c \leftarrow (\mathcal{B}_3[p] \ll 2) \lor (\mathcal{B}_2[p] \ll 1) \lor \mathcal{B}_1[p]$ 
              \STATE $e \leftarrow e_{base} + c$ \hfill \textcolor{blue}{$\triangleright$ Implicit Lookup}
              \STATE $w_k \leftarrow \textsc{MakeBF16}(val.sign, e, val.mantissa)$
          \ELSE
              \STATE \textcolor{blue}{$\triangleright$ Case B: Fallback Path}
              \STATE $idx_{\mathcal{L}} \leftarrow p - idx_{\mathcal{H}}$ \hfill \textcolor{blue}{$\triangleright$ Calculate index in fallback}
              \STATE $w_k \leftarrow \mathcal{L}[\text{start}_{\mathcal{L}} + idx_{\mathcal{L}}]$
          \ENDIF
      \ENDFOR
      
      \STATE $R \leftarrow \textsc{PackRegister}(w_0, w_1)$
      \STATE \textbf{return} $R$
  \end{algorithmic}
\end{algorithm}

\PHM{Spatial Bitmap Indicator.} Each thread first determines the storage mode of its assigned elements by evaluating a spatial indicator mask. During offline compression, each $8 \times 8$ weight tile is encoded using three 64\texttt{-}bit bitmaps, where each bitmap encodes a single bit of the 3\texttt{-}bit codeword. At runtime, the three bitmaps are combined using a warp\texttt{-}level bitwise \texttt{OR} to produce a single 64\texttt{-}bit indicator mask. Each bit in this mask specifies the storage mode of one element: 1 for compressed (high\texttt{-}frequency), 0 for fallback (uncompressed). Each thread determines its decoding path by inspecting the corresponding bits in this spatial indicator mask, which resides in registers. Specifically, for thread~$i$, the bits at positions $2i$ (for $a_0$) and $2i+1$ (for $a_1$) indicate the state of the two assigned elements. For instance, Thread~19 finds that bit~38 ($2 \times 19$) is set, indicating its $a_0$ element is stored in compressed form. It fetches the packed value from the high\texttt{-}frequency buffer and proceeds with exponent reassembly. In contrast, Thread~6 sees that bit~12 ($2 \times 6$) is unset and simply loads its $a_0$ directly from the fallback buffer. This bitwise decision process is lightweight, fully register\texttt{-}resident, and completes in constant time.

\PHM{Dynamic Addressing.} 
Once the storage mode is determined, each thread computes its read offset into the appropriate value buffer on-the-fly, without explicit per-element indices. This is achieved via a lightweight, warp-local prefix sum over the spatial indicator. For thread $i$, the offset is calculated by counting how many previous elements of the same storage type appear in bits $[0,\, 2i - 1]$ of the spatial indicator. Specifically, if the element is compressed ($\text{bit} = 1$), the offset equals the number of $1$s; if uncompressed ($\text{bit} = 0$), it equals the number of $0$s in that range. These counts are efficiently computed using GPU-native instructions such as \texttt{\_\_popc()} and \texttt{\_\_shfl\_sync()}. For example, Thread~6, encountering an unset bit at position~12, computes its fallback buffer offset by counting the number of $0$s in bits $[0,\, 11]$. Thread~19, with bit~38 set, counts the number of $1$s in bits $[0,\, 37]$ to access the compressed buffer. This dynamic addressing mechanism transforms indexing into a deterministic, SIMT-friendly operation that aligns naturally with GPU execution patterns.

\PHM{Fast Exponent Reassembly via Implicit Lookup.} To further reduce the decoding overhead, \SystemName{} reconstructs exponents using an \textit{implicit lookup} mechanism based on \textit{arithmetic remapping}, avoiding table-based decoding. During offline compression, the top-7 most frequent exponent values are identified globally and assigned 3-bit codewords (\texttt{001}–\texttt{111}), ordered by increasing numerical value instead of frequency rank. A single global base exponent is recorded as $\texttt{base\_exp} = \min(\texttt{top\_exponents}) - 1$, which is shared by all tiles. At runtime, each thread reconstructs the original exponent by adding the 3\texttt{-}bit codeword to the base exponent. This operation eliminates shared memory table lookups by using a single integer ALU instruction. The recovered exponent is then fused with the sign and mantissa fields to assemble a valid BF16 value. For example, Thread~19 observes that bit~38 in the spatial indicator is set and reconstructs the 3\texttt{-}bit codeword by reading the corresponding bits from the three bitmap planes, yielding \texttt{101} (5). With a global base exponent of~115, it recovers the original exponent as $115 + 5 = 120$, then combines it with the sign and mantissa to form the final BF16 value. This arithmetic decoding process is fully SIMT-compatible, exploits the GPU’s integer pipelines.

\PHM{Repacking into Tensor Core Fragments.} Each thread repacks the two reconstructed BF16 elements into a single \texttt{bfloat162} register, matching the operand layout required by Tensor Core \texttt{mma.sync} instructions.

\subsubsection{Fine-grained Software Pipeline}

ZipGEMM uses a hierarchical two-level pipeline to overlap memory transfer, decompression, and computation, effectively hiding memory and decompression latency. At the coarse level, tile-wise double buffering overlaps global-to-shared memory transfers with computation; at the fine level, slice-wise interleaving overlaps shared-to-register movement and decompression with Tensor Core operations. This is implemented via two shared memory buffers for compressed weights (triple bitmaps, packed sign-mantissa, fallback values) and activations. Within each tile, computation is sliced along the K dimension (typically $16 \times 16$ fragments) and processed using an interleaved load-decompress-compute pattern. While Tensor Cores execute matrix multiplication (\texttt{mma}) on slice $i$, ALU units concurrently load and decompress weights for slice $i+1$ from shared memory into registers. This ensures a steady compute flow by hiding decompression and memory latency behind computation.

To coordinate the two pipeline levels, ZipGEMM uses a hierarchical barrier strategy for inter-tile and intra-warp synchronization. \textbf{Inter-tile synchronization:} \texttt{cp.async.wait\_group<0>()} and \texttt{\_\_syncthreads()} ensure all asynchronous transfers complete before switching buffers. This barrier is placed \emph{after the final slice decompression but before the final slice \texttt{mma}}, allowing computation to proceed while the next tile is being loaded and decompressed, which maximizes overlap and minimizes stalls. \textbf{Intra-warp coordination:} Intra-warp operations are implicitly synchronized via the SIMT model, requiring no explicit barriers between load, decompress, and compute at the slice level.

\begin{figure}[tbp]
  \centering
  \includegraphics[width=1.0\linewidth]{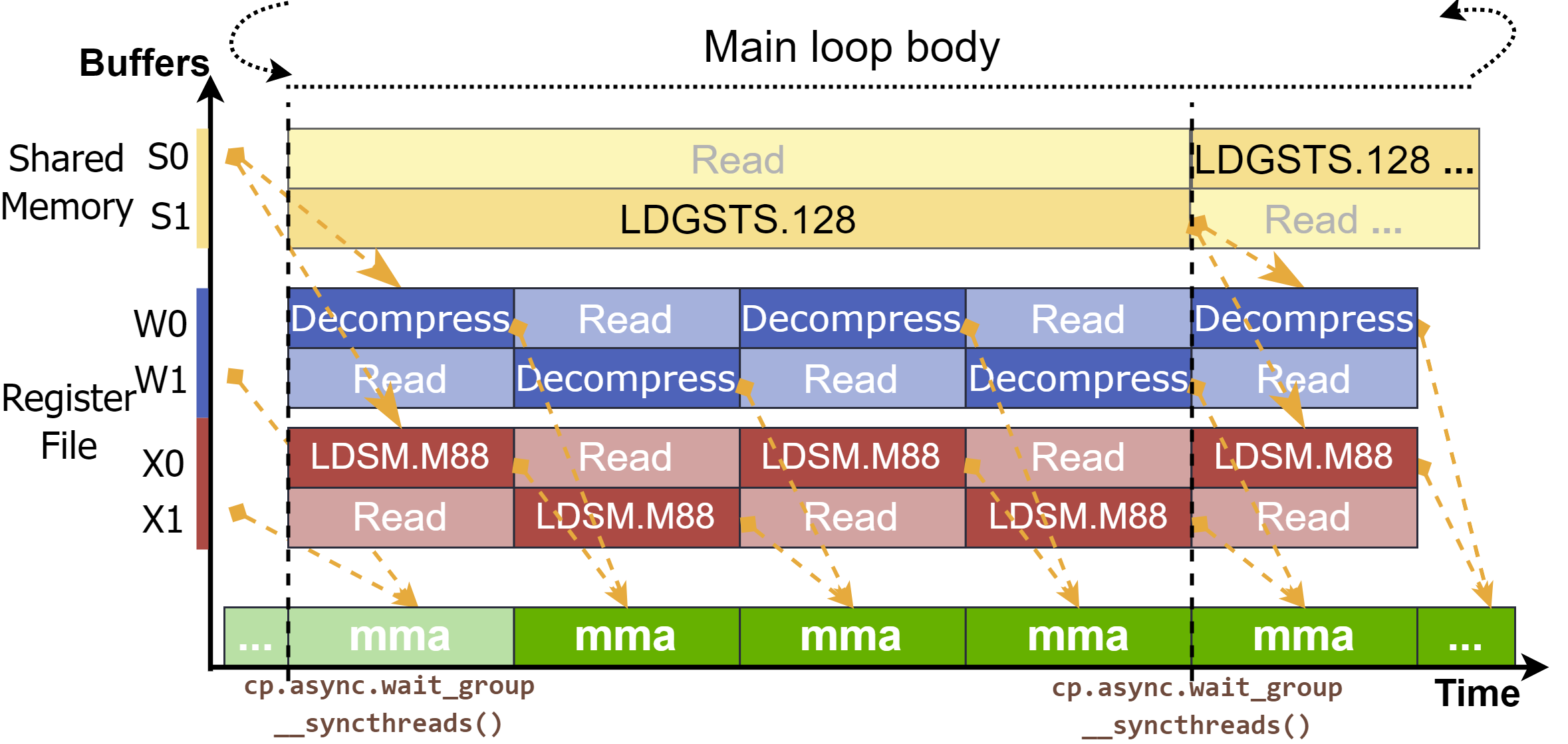}
  \caption{Hierarchical software pipeline design.}
  \label{fig::pipeline}
\end{figure}

\subsection{Stage-Aware Inference Strategy}

\SystemName{} uses the fused ZipGEMM kernel exclusively during the decode stage for accelerated token generation. For the compute-bound prefill stage, where large matrix dimensions ($N=BS\times Seq\_len$) provide high arithmetic intensity, \SystemName{} falls back to a decoupled pipeline: an efficient decompression kernel first extracts the compressed weights to global memory, then performs high-throughput GEMM operations to amortize the decompression overhead (typically $<$4\% as shown in \S\ref{sec::overhead}). In both prefill and decode stages, the decompression kernel and ZipGEMM kernel share the same compressed format and per-thread decompression logic (\S\ref{sec::decomp}), obviating the need for runtime format conversions.

\begin{figure*}[tbp]
  \centering
  \includegraphics[width=0.98\linewidth]{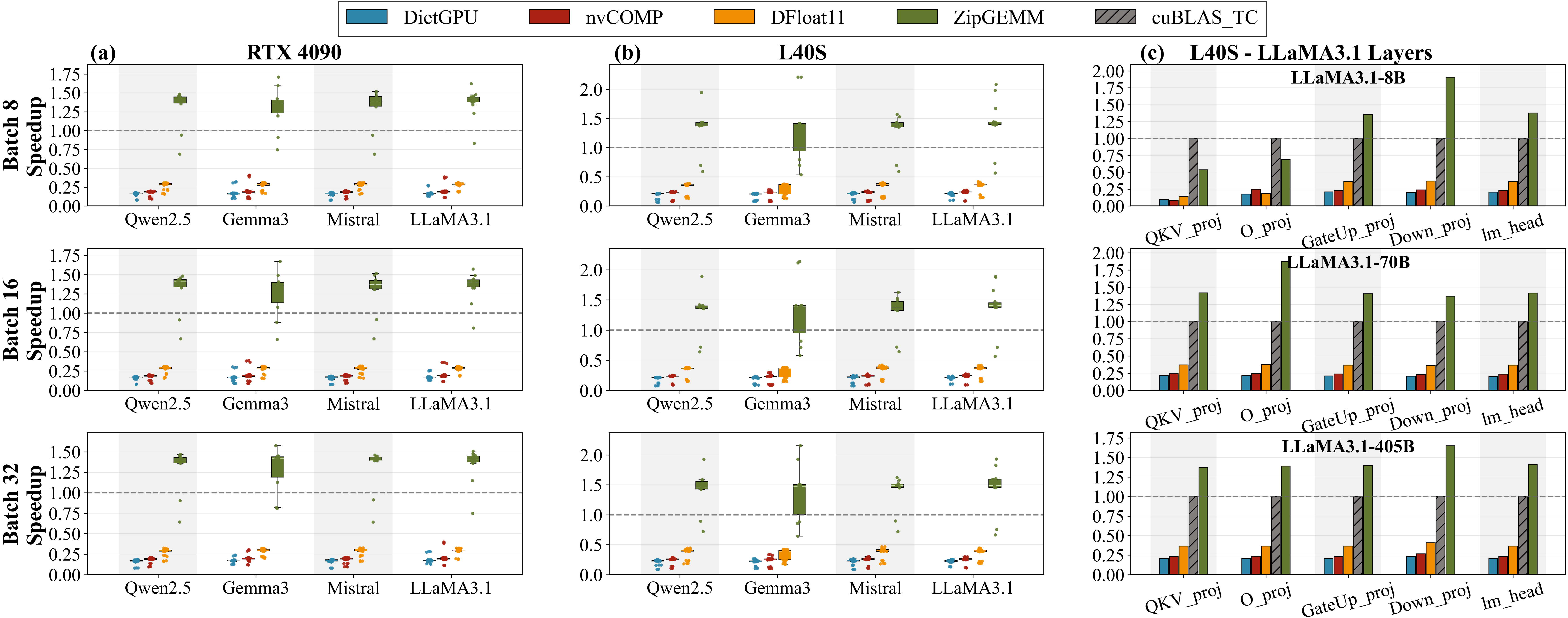}
  \caption{Kernel performance comparison on NVIDIA RTX4090 and L40S GPUs.} 
  \label{fig::speedup}
\end{figure*}

\section{Implementation}
\label{sec:impl}
We implemented \SystemName{} as a high-performance, modular inference backend comprising approximately 3.5K lines of code. The core engine consists of about 2.5K lines of CUDA and C++, which implements the offline TCA-TBE compressor and the online ZipGEMM kernel. The kernel is compiled into a standalone shared library (\texttt{.so}) using \texttt{nvcc}, exposing C++ APIs for weight packing and kernel launching. The remaining 1.0K lines are Python glue code used to integrate \SystemName{} into \textbf{vLLM}~\cite{vllm}. We extended vLLM's model loader and linear execution modules to support the TCA-TBE format, utilizing \texttt{PyBind11} to invoke our custom CUDA kernels.

\section{Evaluation}
We evaluate the performance of \SystemName{} at two levels: the kernel level of the fused ZipGEMM and the standalone Decompression kernel (\SystemName{}-Decomp), and the end-to-end inference framework level. All experiments are conducted on two platforms. \ding{182} A consumer-grade server equipped with 4$\times$ NVIDIA RTX4090 GPUs (Ada Lovelace, 24GB memory, Compute Capability 8.9), paired with an Intel Xeon Platinum 8352V CPU (144 cores, 512GB DDR4). \ding{183} A datacenter platform with 4$\times$ NVIDIA L40S GPUs (Ada Lovelace, 48GB), paired with an Intel Xeon Gold 6230R CPU (104 cores, 512GB DDR4). We also evaluate ZipGEMM on the latest \ding{184} RTX5090 GPU (Blackwell, 32GB, Compute Capability 12.0) to demonstrate forward compatibility. All code is compiled using GCC 11.3 and NVCC 12.4 (with NVCC 12.8 specifically for RTX5090). For kernel-level evaluation, we perform 100 warm-up iterations followed by 1,000 timed executions. For end-to-end evaluation, each configuration is run 10 times.

\subsection{ZipGEMM Kernel Performance}
\PHM{Datasets.} We benchmark the kernel-level performance
on representative linear layers from state-of-the-art LLMs. The input shapes for kernel benchmarking are directly extracted from the real weight matrices of prominent LLM families, including LLaMA3.1~\cite{dubey2024llama} (8B, 70B, and 405B), Qwen2.5~\cite{yang2024qwen2.5} (7B, 14B, 32B, and 72B), Gemma3~\cite{gemma3} (12B and 27B), and Mistral~\cite{mistral} (24B and 123B), covering a broad range of model scales and hidden dimensions.

\PHM{Baselines.} We compare ZipGEMM against four representative baselines: \ding{182} cuBLAS\_TC v12.4.5~\cite{NVIDIA2024cuBLAS}, NVIDIA’s official BF16 Tensor Core GEMM kernel; \ding{183} DietGPU~\cite{dietgpu}, a popular open-source, GPU-native rANS codec for lossless decompression of floating-point weights; \ding{184} nvCOMP (rANS)~\cite{NVIDIA_nvcomp_2025}, NVIDIA’s general-purpose asymmetric numeral systems-based decompression library; and \ding{185} DFloat11~\cite{dfloat11}, a state-of-the-art Huffman-coded GPU decompression framework for LLM inference. Since nvCOMP lacks native BF16 support, we compress exponent bits as a bitstream via rANS and reconstruct BF16 values with a custom high-performance kernel. For DFloat11, whose compression code is unavailable, we benchmark full Transformer block decompression latency and linearly scale estimates for other matrix shapes.

\PHM{Workloads.} We profile all linear layers within a Transformer block, including the merged QKV projection (\texttt{QKV\_proj}), attention output projection (\texttt{O\_proj}), merged FFN gate and up projection (\texttt{GateUp\_proj}), and down projection (\texttt{Down\_proj}), along with the model's LM head layer. Benchmarks are conducted at batch sizes of 8, 16, and 32.

\PHM{Results.} We begin by evaluating the performance of our fused \textit{ZipGEMM} kernel. Figure~\ref{fig::speedup} shows the normalized speedup relative to cuBLAS\_TC across all evaluated models and workloads. ZipGEMM consistently outperforms all baseline methods on both hardware platforms. On the RTX4090, ZipGEMM achieves an average speedup of $1.31\times$ over cuBLAS\_TC, with a peak speedup of $1.71\times$. The advantage is even greater on the L40S, with an average speedup of $1.36\times$ and a maximum of $2.21\times$. In contrast, other decoupled decompression methods introduce substantial overhead, resulting in significant slowdowns. Specifically, DietGPU, nvCOMP, and DFloat11 achieve average speedups of only $0.17\times/0.20\times$, $0.19\times/ 0.23\times$, and $0.28\times/ 0.34\times$ on RTX4090 and L40S, respectively. This indicates that the decoupled decompression processes incur overheads that exceed the computation time of the baseline GEMM. ZipGEMM stands out as the only implementation that can significantly surpass the efficient Tensor Core GEMM. These results highlight the effectiveness of ZipGEMM's fused decompression-computation approach, which efficiently transforms storage savings into tangible execution speedup. 

We further conducted a layer-wise analysis (Figure~\ref{fig::speedup}(c)). ZipGEMM exhibits significant acceleration on most of the computationally intensive layers within a transformer block. For instance, within the LLaMA3.1 model family on the L40S, ZipGEMM achieves average speedups of $1.39\times$ and $1.64\times$ on the GateUp\_proj and Down\_proj layers, respectively. However, ZipGEMM may experience a slowdown when processing certain layers with small shapes; for example, on the L40S, its performance on the O\_proj layer of LLaMA3.1-8B is reduced to $0.79\times$. This is primarily because small layers require fine-grained parameter tuning (e.g., split-K configurations and precise tiling) to fully utilize hardware, which is beyond the scope of this work. Nevertheless, such layers account for only a small fraction of the total FLOPs within a Transformer block. ZipGEMM delivers robust block-level speedups of $1.35\times$ for LLaMA3.1-8B and $1.48\times$ for LLaMA3.1-405B on the L40S. 

\begin{figure}[tbp]
  \centering
  \includegraphics[width=0.98\linewidth]{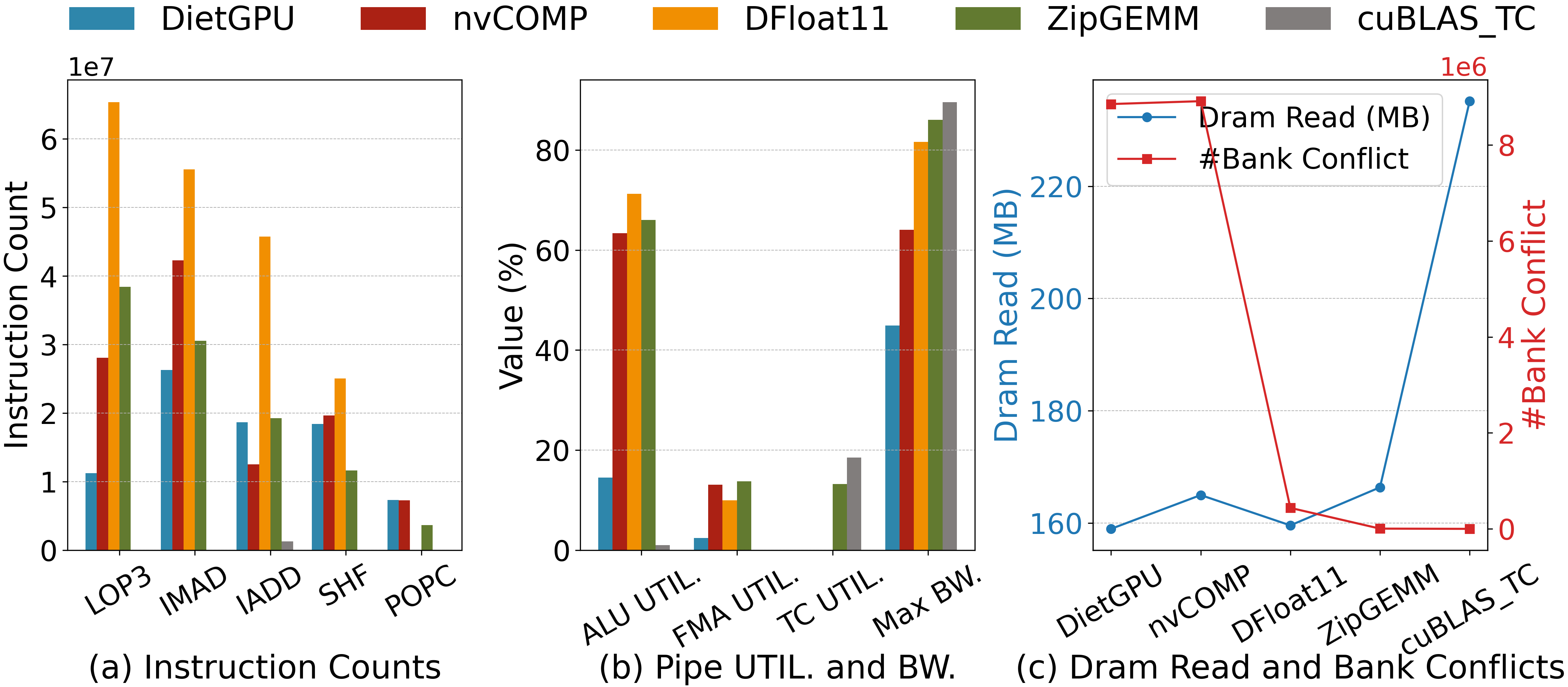}
\caption{Micro-level kernel performance analysis.}
\label{fig::profile}
\end{figure}
\PHM{Micro-level Analysis.} We profiled ZipGEMM with Nsight Compute (NCU) on an RTX4090 to identify the source of its speedup ($M=28672, K=4096$ and $N=32$). As shown in Figure~\ref{fig::profile}, the performance gain stems from a deliberate architectural trade-off: introducing a predictable ALU workload for on-the-fly decoding in exchange for a reduction in memory traffic. Figure~\ref{fig::profile}(a) quantifies this trade-off. The high volume of integer and logical instructions (LOP3, IADD, and POPC) reflects the computational cost of our core decoding steps. This workload is the price for a 29.3\% drop in DRAM reads, a direct validation of the TCA-TBE format's efficiency. Crucially, the two-level software pipeline effectively hides the decoding latency by overlapping it with compute and memory operations. As a result, even with ALU utilization soaring to 66.0\%, Tensor Core utilization is maintained at a remarkable 71.6\% of the cuBLAS baseline, demonstrating that compute throughput is preserved (Figure~\ref{fig::profile}(b)). This high pipeline efficiency is enabled by our data layout. As seen in Figure~\ref{fig::profile}(c), shared memory bank conflicts are virtually eliminated ($\sim$4.7K) compared to the millions incurred by methods like DietGPU. This conflict-free access is a prerequisite for our fine-grained pipeline, ensuring smooth data flow and maximizing SIMT throughput.

\begin{figure}[tbp]
  \centering
  \includegraphics[width=1.0\linewidth]{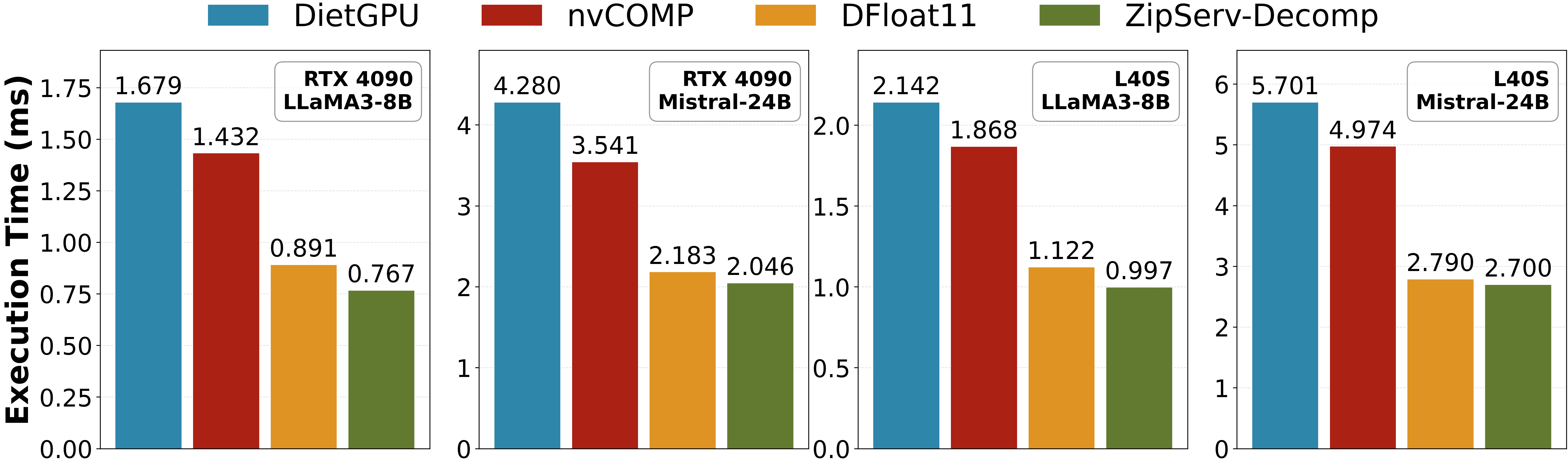}
\caption{Standalone decompression kernel comparison.} 
\label{fig::vsdecomp}
\end{figure}

\subsection{Decompression Kernel Performance}
To further dissect decompression efficiency, we benchmark our standalone \SystemName{}-Decomp kernel. Figure~\ref{fig::vsdecomp} presents the total decompression time for all weights in a full Transformer block of LLaMA3.1-8B and Mistral-24B. \SystemName{}-Decomp achieves average speedups of \textbf{2.14$\times$}, \textbf{1.83$\times$}, and \textbf{1.10$\times$} over DietGPU, nvCOMP, and DFloat11, respectively. Although the TCA-TBE format was co-designed to support fused execution with matrix multiplication, its structure proves highly efficient for standalone decompression as well. This efficiency stems from its fixed-length, warp-aligned design, which eliminates control divergence and enables warp-synchronous per-thread decoding. In contrast, although existing baselines are explicitly optimized for decompression, they often rely on variable-length, entropy-coded formats. These lead to thread divergence, serialized bit parsing, and irregular memory access that degrade GPU efficiency.

\begin{figure}[tbp]
  \centering
  \includegraphics[width=1.0\linewidth]{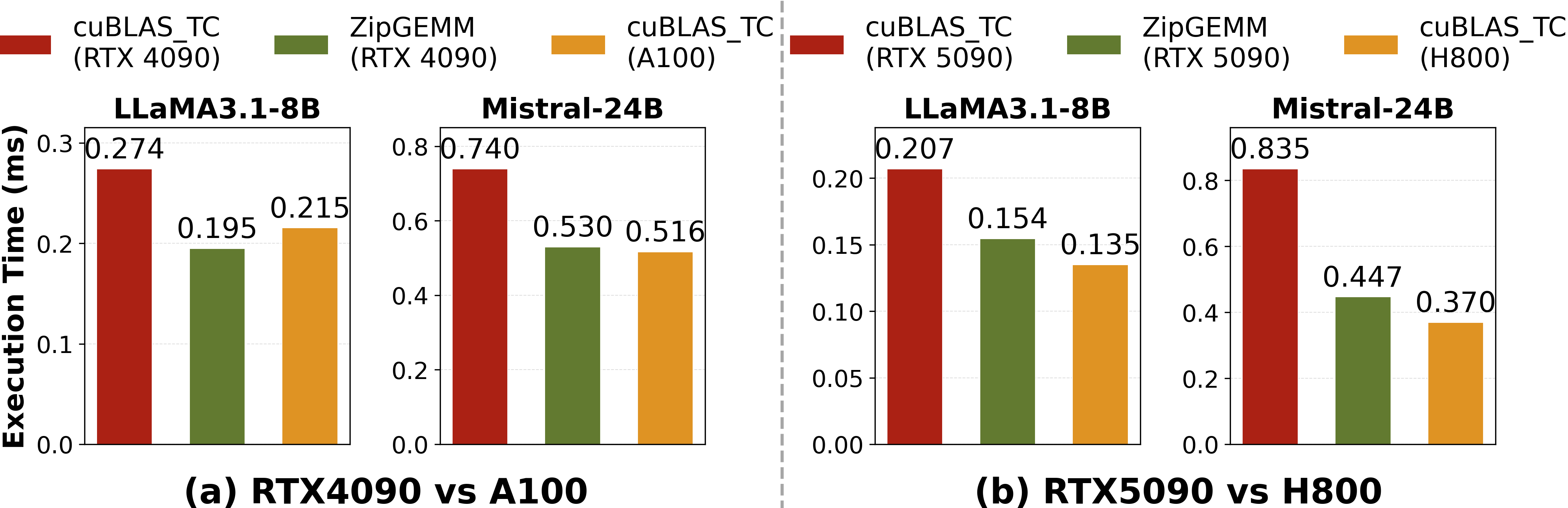}
\caption{Cross-generation performance comparison.}
\label{fig::vshighend}
\end{figure}

\subsection{Performance Across GPU Generations and Tiers}\label{sec::vshighend}
To establish forward compatibility, we benchmark ZipGEMM on the latest NVIDIA RTX5090 and compare it against top-tier datacenter A100 and H800 using LLaMA3.1-8B and Mistral-24B GateUp\_proj layers at batch size 32. We first directly port ZipGEMM to the Blackwell-based RTX5090 without exploiting new features (e.g., Tensor Memory and asynchronous \texttt{WMMA} execution~\cite{jarmusch2025dissecting}). As shown in Figure~\ref{fig::vshighend}, ZipGEMM delivers substantial speedups over cuBLAS\_TC on RTX5090—1.34$\times$ for LLaMA3.1-8B and 1.87$\times$ for Mistral-24B—confirming the design to be forward-compatible. ZipGEMM also narrows the consumer–datacenter divide: on an RTX4090, ZipGEMM outperforms the standard cuBLAS\_TC on A100 with LLaMA3.1-8B (0.195 ms vs. 0.215 ms, 9.3\% faster) and is only 2.7\% slower on Mistral-24B (0.530 ms vs. 0.516 ms), effectively placing it in the same performance class. This trend intensifies on newer hardware. While a standard RTX5090 trails the H800 by 53.3\% (LLaMA3.1-8B) and 125.7\% (Mistral-24B), ZipGEMM reduces these deficits to 14.1\% and 20.8\%, respectively (Figure~\ref{fig::vshighend}(b)), approaching datacenter-level performance on consumer GPUs.

\begin{figure}[tbp]
  \centering
  \includegraphics[width=1.0\linewidth]{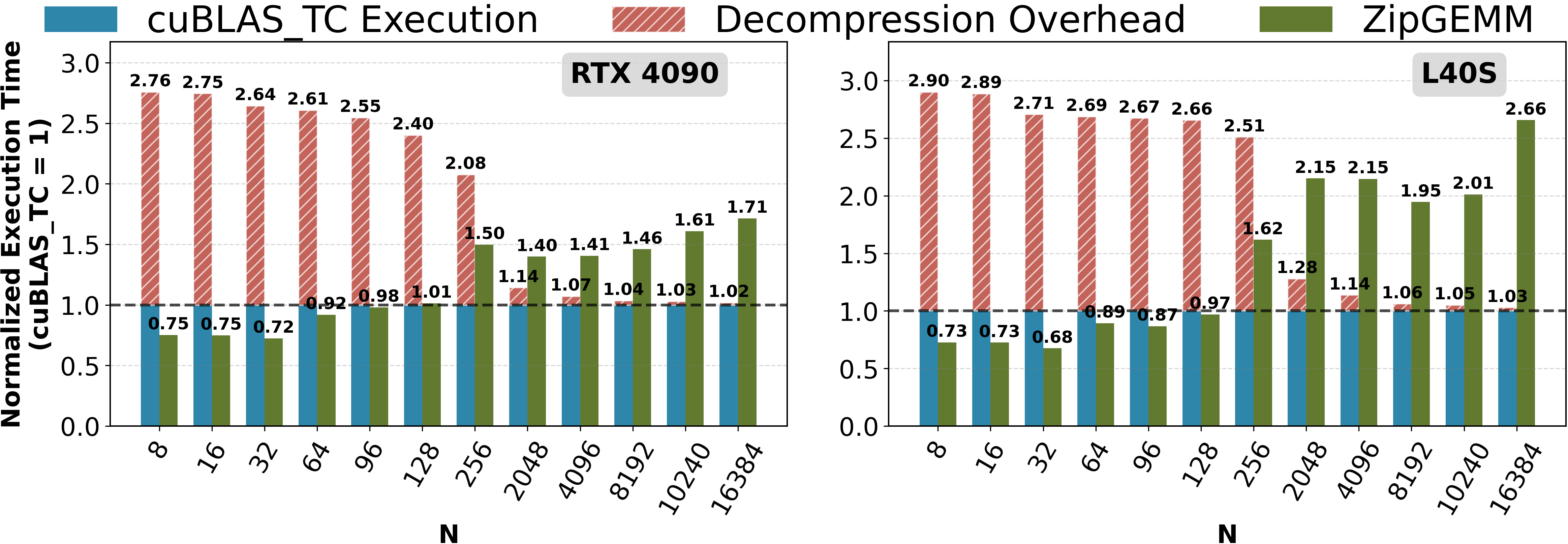}
\caption{\SystemName{} performance under different N settings.}
\label{fig::limitation2}
\end{figure}
\subsection{Overhead Analysis}\label{sec::overhead}
We analyze the system overhead from two perspectives: runtime inference overhead and offline preparation cost. \ding{182} Runtime Overhead. Figure~\ref{fig::limitation2} quantifies the overhead of \SystemName{} during inference across different $N$ settings ($N=BS \times Seqlen$). In the decode stage (small $N$, typically 1--128), the fused ZipGEMM kernel incurs no overhead. Instead, it consistently outperforms the cuBLAS\_TC baseline in these memory-bound regimes, with on-the-fly decompression fully hidden within the kernel execution. For the compute-bound prefill stage (large $N$, e.g., 8192), where ZipGEMM's on-the-fly decompression overhead outweighs its benefits from reduced memory access, \SystemName{} switches to a decoupled pipeline. The efficient decompression kernel first expands the compressed weights, followed by cuBLAS\_TC GEMMs. This incurs a limited overhead of only $\sim$4\%/2\% of the GEMM time at $N=8192/16384$. \ding{183} Offline Compression Cost. Beyond runtime performance, we also evaluate the one-time cost of preparing the model. Compressing the LLaMA-3.1-8B model takes approximately 2.5 minutes on a 16-core Intel Xeon 8352V CPU. Given that this is an offline operation performed only once prior to deployment, it does not impact the critical path of online serving and is negligible when amortized over the model's lifecycle.

\subsection{End-to-end Inference Performance}


\PHM{Setup.} We evaluate the end-to-end inference performance of \textsc{ZipServ} on a range of representative models and hardware configurations: LLaMA3.1-8B on one RTX4090 GPU, Mistral-24B on two L40S GPUs, and LLaMA3.1-70B on four L40S GPUs with tensor parallelism. We benchmark using batch sizes of 8 and 32, with varied output sequence lengths of 128, 256, 512, 1024, and 2048 tokens to simulate different serving scenarios. We compare \textsc{ZipServ} against three leading baseline systems: \ding{182} vLLM~\cite{vllm}, a state-of-the-art LLM inference and serving framework; \ding{183} Transformers~\cite{wolf2020transformers}, a widely adopted standard library; and \ding{184} DFloat11~\cite{dfloat11}, representing state-of-the-art performance for lossless compression-based inference frameworks. We measure two key metrics: end-to-end request latency (total time to generate the full output sequence) and throughput (output tokens per second). As shown in Figure~\ref{fig::e2e}, \textsc{ZipServ} consistently demonstrates superior performance across all tested configurations.

\begin{figure}[tbp]
\centering
\includegraphics[width=1.0\linewidth]{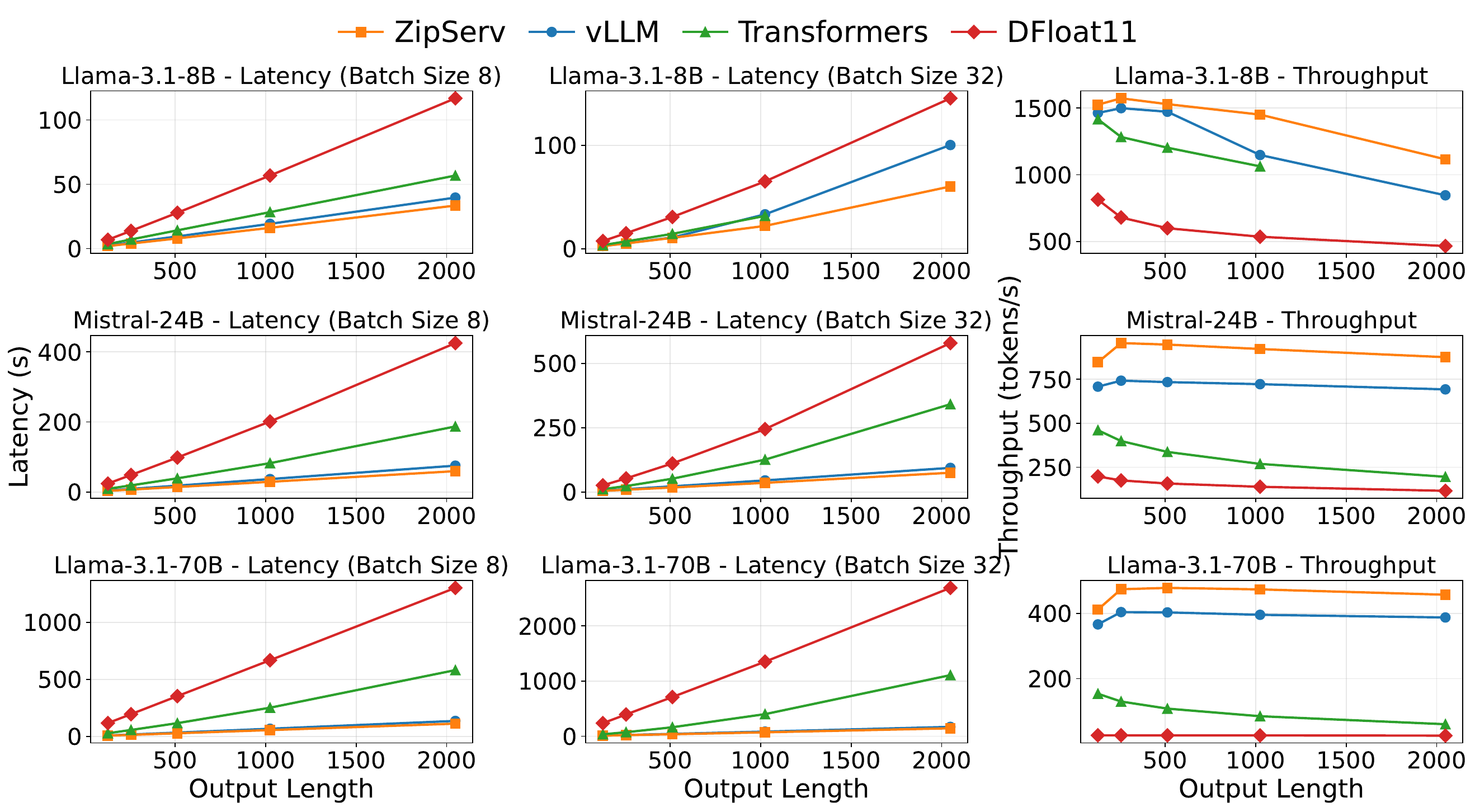}
\caption{End-to-end performance comparison.}
\label{fig::e2e}
\end{figure}

\PHM{Results.} For \textbf{latency}, on average, across all models and batch sizes, \textsc{ZipServ} reduces latency by 17.60\%, 60.79\%, and 82.13\% compared to vLLM, Transformers, and DFloat11, respectively. For \textbf{throughput}, \textsc{ZipServ} provides average speedups of 1.22$\times$ over vLLM, 3.18$\times$ over Transformers, and 8.52$\times$ over DFloat11. The performance gains are pronounced for long-context generation, where the memory-bandwidth savings and computational efficiency of the fused ZipGEMM kernel in the decode phase become dominant. For instance, when generating 2048 output tokens with batch size of 32 using LLaMA3.1-8B, \textsc{ZipServ} achieves a throughput of 1105 tokens/sec, resulting in a 1.66$\times$ speedup over vLLM. We also analyzed the \textbf{memory consumption} during inference. For LLaMA3.1-8B, Mistral-24B, and LLaMA3.1-70B, \textsc{ZipServ} reduces the weight footprint of 14.96/43.92/131.56 GB down to 10.83  (72.4\%)/31.30 (71.3\%)/93.52 (71.1\%) GB, respectively. The reduction in weight storage further enhances serving efficiency in two key ways. First, it enables the deployment of larger models on resource-constrained hardware. Second, the freed memory can be allocated to the KV cache, allowing memory managers like vLLM's PagedAttention~\cite{vllm} to support larger batch sizes and longer contexts, thereby converting static weight savings into dynamic throughput gains.

\PHM{Breakdown Analysis.} We further dissect the performance gains by analyzing the latency and memory composition of LLaMA-3.1-8B on an RTX4090, as detailed in Figure~\ref{fig::limitation}. In the baseline vLLM system (at sequence length 1024), GEMM operations dominate the runtime, consuming 24.99 ms (83.6\% of total latency). \SystemName{} effectively alleviates this bottleneck: the fused ZipGEMM kernel, combined with residual dense GEMMs, reduces the total linear layer latency to 14.76 ms, a 1.69$\times$ improvement. Since Attention (3.02 ms) and other overheads (1.88 ms) remain constant, these kernel-level gains directly drive the end-to-end speedup. On the memory front, \SystemName{} compresses the static weights from 14.96 GB to 11.18 GB. This 3.78 GB saving is automatically repurposed by the memory manager to expand the KV cache capacity from 5.07 GB to 8.60 GB (a 1.70$\times$ increase), thereby enabling the higher throughput and longer context support observed in our end-to-end benchmarks.

\begin{figure}[tbp]
  \centering
  \includegraphics[width=1.0\linewidth]{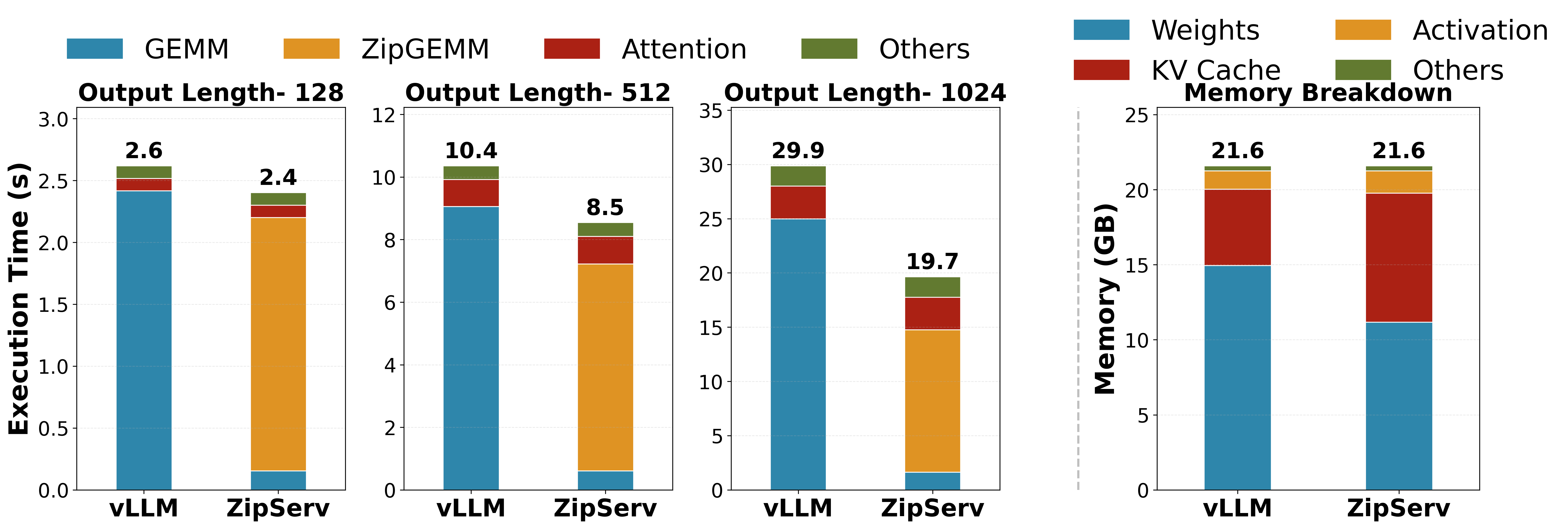}
\caption{Breakdown of end-to-end inference time and memory consumption.}
\label{fig::limitation}
\end{figure}

\begin{figure}[tbp]
  \centering
  \includegraphics[width=1.0\linewidth]{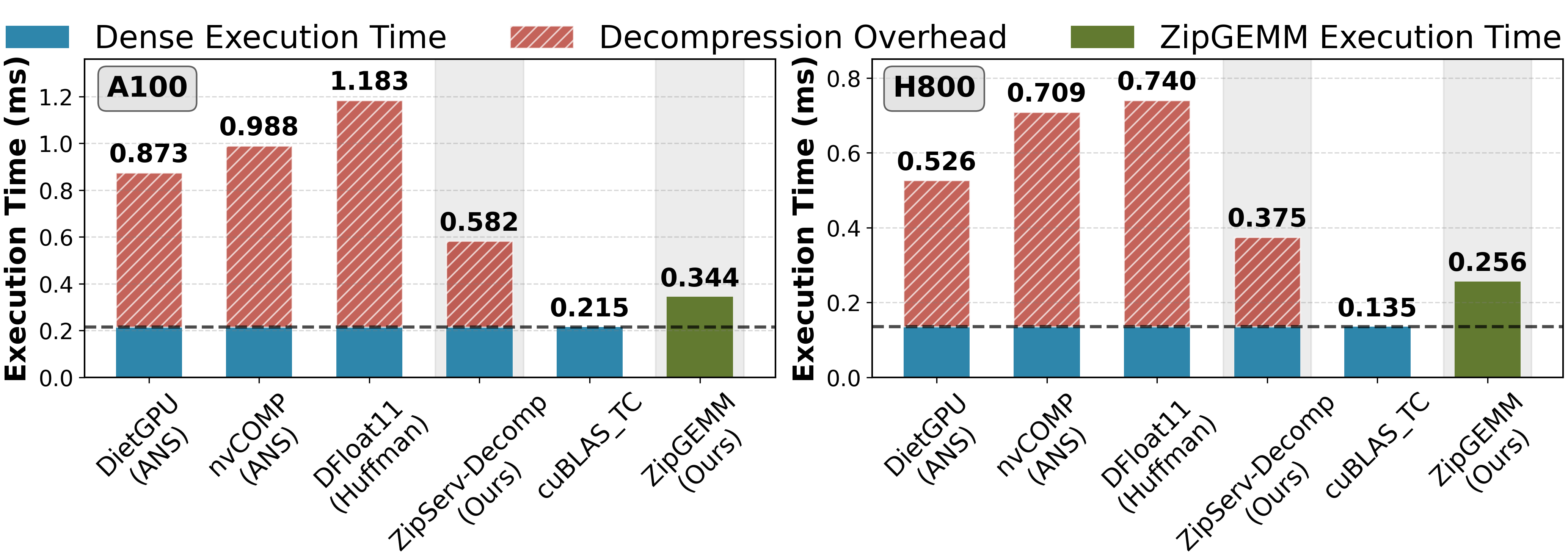}
\caption{Performance on training-oriented GPUs.}
\label{fig::limitation}
\end{figure}

\section{Limitation and Discussion}\label{sec:discussion}
\SystemName{} is designed for the increasingly important deployment scenario on resource-constrained consumer-grade and inference-optimized GPUs, where limited memory bandwidth makes lossless compression a powerful lever for efficiency. On such platforms, \SystemName{} consistently delivers substantial acceleration and memory savings.  To stress-test performance under more bandwidth-relaxed conditions, we also benchmarked on training-oriented datacenter GPUs (A100, H800), where ZipGEMM may not always match the highly optimized cuBLAS baseline (Figure~\ref{fig::limitation}). This reflects a hardware--software mismatch rather than an algorithmic limitation: abundant HBM (HBM2e/HBM3) alleviates the memory bottlenecks \SystemName{} is designed to mitigate, while lower core frequencies (e.g., 1410 MHz on A100 vs. 2520 MHz on RTX4090) make the intensive ALU workload harder to hide within the software pipeline.   Nevertheless, \SystemName{} still provides best-in-class support for compressed inference. Our standalone decompression kernel outperforms state-of-the-art by up to 2.64$\times$, and ZipGEMM remains the fastest fused GEMM kernel. As shown in \S\ref{sec::vshighend}, \SystemName{} also enables consumer-grade GPUs to close much of the gap with elite datacenter accelerators, offering a compelling cost-performance proposition for deployment on accessible hardware. 

While \SystemName{} targets bit-exact inference, a comparison with lossy techniques is instructive. ZipGEMM was benchmarked against the Marlin W8A16 FP8 kernel on an RTX4090 GPU, using a representative weight shape ($28672 \times 4096$) at batch size 32. Although ZipGEMM trails Marlin-W8A16 in latency (0.194 ms vs. 0.143 ms), the resulting $1.36\times$ gap aligns closely with the ratio of effective bit-widths ($\sim$11 bits vs. FP8). This indicates that our design reduces and hides the overhead of complex lossless decompression within the memory access latency. Furthermore, \SystemName{} is orthogonal to lossy methods and can be applied atop quantized weights to exploit residual redundancy, combining aggressive compression with enhanced performance~\cite{gerogiannis2025deca}.

Three key directions are envisioned for extending \SystemName{}. First, the TCA-TBE format can be adapted for lossless KV Cache compression, addressing the dominant memory bottleneck in long-context serving~\cite{liu2024cachegen}. Second, although currently optimized for NVIDIA architectures, ZipGEMM can be adapted to other matrix accelerators, including Intel AMX~\cite{kim2024exploiting} and AMD Matrix Cores~\cite{schieffer2024rise}. This extensibility is supported by the hardware-agnostic nature of the core design, as the integer arithmetic and population count instructions required for decompression are widely supported across modern instruction sets. Finally, \SystemName{} is applicable to broader system-level challenges, including efficient model checkpointing ~\cite{waddington2025lossless,strati2025pccheck} and communication compression in distributed training~\cite{wang2025zen,zhang2023evaluation}.

\section{Related Work}
\PHM{Lossy Model Compression.} Lossy methods dominate LLM acceleration, mainly via post-training quantization (PTQ)~\cite{frantar2022gptq, dettmers2022gpt3, xiao2023smoothquant, lin2024awq, ashkboos2024quarot, zhao2024atom, liu2024spinquant, ptq_oac, dong2024stbllm} and pruning~\cite{Frantar2023SparseGPTML,Sun2023ASA_wanda,dong2024pruner, das2023beyond,xu2024besa,zhang2024plugandplay}, supported by efficient kernels~\cite{QIGen, nuQmm, frantar2024marlin, lut_gemm, wang2024ladder, SpInfer, xia2023flash}. These approaches risk accuracy degradation~\cite{Quantization_Degradation,dong2025can}. \SystemName{} provides bit-exact, lossless, and efficient inference.

\PHM{Lossless Model Compression.} A large body of work has investigated memory compression to reduce bandwidth or expand capacity via lightweight hardware schemes~\cite{kim2016bit, choukse2020buddy, ekman2005robust, zhao2015buri, pekhimenko2013linearly, choukse2018compresso,pekhimenko2012base}, but these techniques are not tailored for model compression. Efforts such as LMC~\cite{waddington2025lossless} and ZipNN~\cite{zipnn} apply Huffman~\cite{huffman2007method} to compress checkpoints for efficient storage and distribution, but offer no runtime benefits. Recent systems, including NeuZip~\cite{neuzip}, DietGPU~\cite{dietgpu}, nvCOMP~\cite{NVIDIA_nvcomp_2025}, and DFloat11~\cite{dfloat11}, support lossless GPU codecs at runtime to reduce inference memory usage, but incur significant overhead (\S\ref{sec:gaps}). Huff-LLM~\cite{yubeaton2025huffllm} achieves higher efficiency but targets FPGA-like architectures and does not generalize to GPUs. Ecco~\cite{cheng2025ecco} designs specialized Huffman codec hardware, but targets lossy compression.  
Our \SystemName{} fuses decompression and GEMM computation, turning lossless compression into practical GPU inference acceleration.

\PHM{Kernel Fusion.} Kernel fusion reduces memory traffic by combining operators, as in FlashAttention~\cite{dao2022flashattentionfastmemoryefficientexact, FlashAttention_3, dao2023flashattention2fasterattentionbetter} or graph-level frameworks~\cite{ma2020rammer, xing2022bolt, wu2025mirage}. \SystemName{} draws insights from them and, to our knowledge, is the first to fuse decompression with GEMM, avoiding full-weight materialization.

\PHM{System-Level Optimizations for LLM Inference.} Modern LLM serving is powered by sophisticated inference engines~\cite{mukherjee2023orcaprogressivelearningcomplex, vllm, song2024powerinfer, gong2025past, zheng2024sglang, zhong2024pd, amey2024sarathi, holmes2024deepspeedfastgen, ServerlessLLM, AlpaServe, Llumnix}, which focus on high-level scheduling strategies and memory orchestration. \SystemName{} is orthogonal and complementary, and can be integrated as a high-performance backend. This allows these engines to benefit from both a reduced memory footprint and accelerated computation without altering their core logic.

\section{Conclusion}
We presented \textbf{\SystemName{}}, a lossless compression framework that, for the first time, delivers significant inference acceleration for Large Language Models. By co-designing a hardware-aware compression format, TCA-TBE, with a fused decompression-GEMM kernel, \SystemName{} overcomes the architectural bottlenecks that have historically plagued lossless methods on GPUs. Our evaluation demonstrates substantial speedups over highly-optimized baselines like cuBLAS, particularly on consumer-grade hardware where \SystemName{} narrows the performance gap to expensive datacenter GPUs, establishing a compelling cost-performance proposition. Ultimately, \SystemName{} reframes lossless compression from a mere storage-saving utility into a practical and powerful tool for high-performance, bit-exact LLM inference.

\begin{acks}
We extend our thanks to the anonymous ASPLOS reviewers and our shepherd, Bo Wu, for their valuable feedback and support. This work was partially supported by National Natural Science Foundation of China under Grant No. 62272122, the Guangzhou Municipal Joint Funding Project with Universities and Enterprises under Grant No. 2024A03J0616, Guangzhou Municipality Big Data Intelligence Key Lab (2023\allowbreak A03J0012), Hong Kong CRF grants under Grant No. C7004-22G and C6015-23G, the NSFC/RGC Collaborative Research Scheme under the contract of CRS\_HKUST601/24, and National Natural Science Foundation of China under Grant No. 62302126. Wei Wang and Xiaowen Chu are the corresponding authors.
\end{acks}
\appendix
\section{Theoretical Analysis: Compressibility of LLM BF16 Weights}\label{sec:theory_appendix}
We present the theoretical foundation showing why exponent distributions in LLM weights are highly skewed and exhibit top‑K contiguity.

Following recent studies~\cite{qlora, gauss1, gauss2}, we assume that weights  $w \in \mathbb{R}^D$ in a single layer (vectorized for analysis) follow a zero-mean normal distribution:

$$w \sim \mathcal{N}(0, \sigma^2 I)$$

A non-zero, normal BF16 number $v$ is represented as $v = (-1)^S \times 2^{E - 127} \times (1.m_1...m_7)_2$, where $S$ is the sign bit, $E$ is the 8-bit unsigned integer value of the exponent field, and $(1.m_1...m_7)_2$ is the 7-bit mantissa with an implicit leading 1. The bias for the BF16 exponent is 127.

Let $x = E - 127$ be the actual exponent value. Any number using this specific exponent $E$ will have a magnitude in the range $[2^x, 2^{x+1})$. Our analysis focuses on the probability distribution of this exponent value $x$ (or equivalently, $E$), given that the weights $w$ are drawn from $\mathcal{N}(0, \sigma^2)$. The redundancy arises if this distribution $P(X=x)$ is highly skewed, meaning some exponent values are far more common than others.

The probability of a single weight $w_i$ falling into the magnitude range corresponding to a specific exponent $x$ is:

$$P(X=x) = P(2^x \leq |w_i| < 2^{x+1})$$

Note that this calculation is an approximation. We are calculating the probability of a value falling into the exponent's ideal magnitude range $[2^x, 2^{x+1})$, which simplifies the BF16 quantization process by ignoring rounding effects caused by the 7-bit mantissa. However, this serves as a robust approximation for analyzing the overall exponent distribution.

Given that $w_i \sim \mathcal{N}(0, \sigma^2)$, its Probability Density Function (PDF) is $f(w_i) = \frac{1}{\sqrt{2\pi\sigma^2}} e^{-w_i^2 / (2\sigma^2)}$. The probability is the integral of this PDF over the positive and negative ranges:

$$P_{\sigma}(X=x) = 2 \times \int_{2^x}^{2^{x+1}} \frac{1}{\sqrt{2\pi\sigma^2}} e^{-t^2 / (2\sigma^2)} dt$$

This integral can be expressed using the error function (erf), defined as $\text{erf}(z) = \frac{2}{\sqrt{\pi}} \int_0^z e^{-u^2} du$:

$$P_{\sigma}(X=x) = \text{erf}\left(\frac{2^{x+1}}{\sigma\sqrt{2}}\right) - \text{erf}\left(\frac{2^x}{\sigma\sqrt{2}}\right)$$

\begin{theorem}[]
The function $P(X=x) = \text{erf}\left( \frac{2^{x+1}}{\sigma\sqrt{2}} \right) - \text{erf}\left( \frac{2^x}{\sigma\sqrt{2}} \right)$ is unimodal for $x \in \mathbb{Z}$.
\end{theorem}

\begin{proof}
To prove unimodality, we consider the continuous extension $f(x) = \text{erf}\left( \frac{2^{x+1}}{\sigma\sqrt{2}} \right) - \text{erf}\left( \frac{2^x}{\sigma\sqrt{2}} \right)$ for $x \in \mathbb{R}$. If $f(x)$ is unimodal, then the discrete function $P(X=x)$, which is the evaluation of $f(x)$ at integer points, will also be unimodal. 

Let $u = \frac{2^x}{\sigma\sqrt{2}}$, so that $f(x) = \text{erf}(2u) - \text{erf}(u)$. The derivative of the error function is $\frac{d}{dz} \text{erf}(z) = \frac{2}{\sqrt{\pi}} e^{-z^2}$. Thus, the derivative of $f$ with respect to $x$ is:

$$\frac{df}{dx} = \frac{2}{\sqrt{\pi}} u \ln 2 e^{-u^2} \left( 2 e^{-3u^2} - 1 \right)$$

Let $h(u) = 2 e^{-3u^2} - 1$. Since $\frac{2}{\sqrt{\pi}}$, $u$, $\ln 2$, and $e^{-u^2}$ are all positive for $u > 0$ (as $2^x > 0$), the sign of $\frac{df}{dx}$ is determined solely by $h(u)$.

Setting $h(u) = 0$ gives:

$$2 e^{-3u^2} = 1 \implies e^{-3u^2} = \frac{1}{2} \implies -3u^2 = -\ln 2 \implies u^2 = \frac{\ln 2}{3}$$

Thus, the unique critical point is at $u_0 = \sqrt{\frac{\ln 2}{3}}$.

For $u < u_0$, we have $3u^2 < \ln 2$, so $e^{-3u^2} > \frac{1}{2}$, meaning $h(u) > 0$ and $\frac{df}{dx} > 0$, so $f(x)$ is increasing.

For $u > u_0$, we have $3u^2 > \ln 2$, so $e^{-3u^2} < \frac{1}{2}$, meaning $h(u) < 0$ and $\frac{df}{dx} < 0$, so $f(x)$ is decreasing.

Therefore, $f(x)$ has a single maximum at $u_0$, proving that it is unimodal. Since $P(X=x)$ is the discrete sampling of $f(x)$ at integer values, it follows that $P(X=x)$ is also unimodal.\end{proof}

\begin{theorem}[]Contiguity of Top-K in Unimodal Distributions.\end{theorem}

\begin{proof}
Proof by contradiction: Suppose that the set $\mathcal{X}_K$ of the Top-K most probable values is not contiguous. Then, there exist three integers $x_a < x_c < x_b$ such that: $x_a, x_b \in \mathcal{X}_K$ but $x_c \notin \mathcal{X}_K$.

By the unimodal property, the probability function $P(x)$ first increases and then decreases, so for any $x_c$ between $x_a$ and $x_b$, we have:

$$P(x_c) \geq \min(P(x_a), P(x_b)).$$

Since $x_a$ and $x_b$ are in $\mathcal{X}_K$, they are among the $K$ largest probabilities. Thus, $\min(P(x_a), P(x_b))$ is at least as large as the $K$-th largest probability. Therefore, $P(x_c)$ must also be at least as large as the $K$-th largest probability, meaning $x_c$ should be in $\mathcal{X}_K$.

This contradicts the assumption that $x_c \notin \mathcal{X}_K$. Hence, the Top-K set must be contiguous.
\end{proof}
\bibliographystyle{ACM-Reference-Format}
\balance
\bibliography{references}

\end{document}